%% file: draftArxiv.tex
\documentclass[leqno]{article}

\usepackage[letterpaper]{geometry}

\usepackage{ltexpprt}

\newcommand{\ilie}{\textcolor{magenta}}

\newcommand{\modif}{\textcolor{blue}}

\usepackage{comment}
\usepackage[tbtags]{amsmath}
\usepackage{amssymb,amsfonts}
\usepackage{subfig}
\usepackage{caption}
\usepackage{algorithmic}
\usepackage{graphicx}
\usepackage{textcomp}
\usepackage{xcolor}
\usepackage[hidelinks]{hyperref}
\usepackage[ruled, vlined,algo2e]{algorithm2e}
\usepackage{booktabs}
\usepackage{adjustbox}
\usepackage{pifont}
\usepackage{csquotes}
\usepackage{float}
\usepackage{makecell}
\usepackage{tabularx} 

\newcommand{\algnameone}{\texttt{PRESTO-A}}
\newcommand{\algnametwo}{\texttt{PRESTO-E}}
\newcommand{\algname}{\texttt{PRESTO}} 

\begin{document}

\title{\algname: Simple and Scalable Sampling Techniques for the Rigorous Approximation of Temporal Motif Counts\thanks{
        This work was supported, in part, by MIUR of Italy, 
        under PRIN Project n. 20174LF3T8 AHeAD, 
        and by the U. of Padova under projects ``STARS 2017" and ``SID 2020: RATED-X''.}}
\author{Ilie Sarpe\thanks{Dept.\ of Information Engineering, University of Padova, Italy. Email: \texttt{\{sarpeilie,vandinfa\}@dei.unipd.it}}
    \and Fabio Vandin$^{\text{\textdagger}}$
}

\date{}

\maketitle







\begin{abstract} \small\baselineskip=9pt The identification and counting of small graph patterns, called \emph{network motifs}, is a fundamental primitive in the analysis of networks, with application in various domains, from social networks to neuroscience. Several techniques have been designed to count the occurrences of motifs in static networks, with recent work focusing on the computational challenges provided by large networks. Modern networked datasets contain rich information, such as the time at which the events modeled by the networks edges happened, which can provide useful insights into the process modeled by the network. The analysis of motifs in temporal networks, called \emph{temporal motifs}, is becoming an important component in the analysis of modern networked datasets. Several methods have been recently designed to count the number of instances of temporal motifs in temporal networks, which is even more challenging than its counterpart for static networks. Such methods are either exact, and not applicable to large networks, or approximate, but provide only weak guarantees on the estimates they produce and do not scale to very large networks. In this work we present an efficient and scalable algorithm to obtain rigorous approximations of the count of temporal motifs. Our algorithm is based on a simple but effective sampling approach, which renders our algorithm practical for very large datasets. Our extensive experimental evaluation shows that our algorithm provides estimates of temporal motif counts which are more accurate than the state-of-the-art sampling algorithms, with significantly lower running time than exact approaches, enabling the study of temporal motifs, of size larger than the ones considered in previous works, on billion edges networks. \end{abstract}

\section{Introduction}The identification of patterns is a ubiquitous problem in data mining~\cite{han2011data} and is extremely important for networked data, where the identification of small, connected subgraphs, usually called \emph{network motifs}~\cite{milo2002network} or \emph{graphlets}~\cite{prvzulj2004modeling}, have been used to study and characterize networks from various domains, including biology~\cite{mangan2003structure}, neuroscience~\cite{battiston2017multilayer}, social networks~\cite{ugander2013subgraph}, and the study of complex systems in general~\cite{milo2004superfamilies}. Network motifs have been used as building blocks for various tasks in the analyses of networks across such domains, including anomaly detection~\cite{sun2007graphscope} and clustering~\cite{benson2016higher}.

A fundamental problem in the analysis of network motifs is the counting problem~\cite{bressan2017counting,ahmed2014graph}, which requires to output the number of instances of the given topology defining the motif. This challenging computational problem has been extensively studied, with several techniques designed to count the number of occurrences of simple motifs, such as triangles~\cite{tsourakakis2009doulion,park2014mapreduce,stefani2017triest} or sparse motifs~\cite{de2017tiered}.

Most recent work has focused on providing techniques for the analysis of \emph{large} networks, which have become the standard in most applications.
However, in addition to a significant increase in size, modern networks also feature a richer structure, in terms of the type of information that is available for their vertices and edges~\cite{ceccarello2017clustering}. A type of information that has drawn significant attention in recent years is provided by the temporal dimension~\cite{holme2012temporal,holme_temporal_2019}. In several applications edges are supplemented with \emph{timestamps} describing the time at which an event, modeled by an edge, occurred: for example, in the analysis of spreading processes in epidemics, nodes are individuals, an edge represents a physical interaction between two individuals, and the timestamp represents the time at which the interaction was recorded \cite{peixoto2018change}.

When studying motifs in temporal networks, one is usually interested in occurrences of a given topology whose edge timestamps all appear in a small time span~\cite{holme2012temporal,paranjape_motifs_2017}. 
Discarding the temporal information of the network, i.e. ignoring the timestamps, may lead to incorrect characterization of the system of interest, while the analysis of temporal networks can provide insights that are not revealed when the temporal information is not accounted for~\cite{holme2012temporal}. For example, in a temporal network, a triangle $x \rightarrow y \rightarrow z \rightarrow x$  represents some feedback process on the information originated from $x$ only if the edges occur at increasing timestamps (and the triangle occurs in a small amount of time). This information is revealed only by considering the timestamps, while by restricting to the static network (i.e., discarding edges timestamps) we may, often incorrectly, conclude that initial information starting from $x$ always affects such sequence of events. Motifs that capture temporal interactions, such as the ones we consider, can provide more useful information than static motifs, as shown in several applications, including network classification~\cite{tu2018network} and in the identification of mixing services from bitcoin networks~\cite{wu2020detecting}. Furthermore, while on static networks motifs with high counts are associated with important properties of the dataset (e.g., its domain), the temporal information provides additional insights on the network and the nature of frequently appearing motifs, such as, for example, the presence of bursty activities~\cite{belth2020persistence}. Unfortunately, current techniques do not enable the analysis of large temporal networks, preventing researchers from studying temporal motifs in complex systems from many areas. 

The problem of counting motifs in temporal networks is even more challenging than  its counterpart for static networks, since there are motifs for which the problem is NP-hard for temporal networks while it is efficiently solvable for static networks~\cite{liu_sampling_2019}. 
Current approaches to count motifs in temporal networks are either \emph{exact}~\cite{paranjape_motifs_2017,mackey_chronological_2018}, and cannot be employed for very large networks, or \emph{approximate}~\cite{liu_sampling_2019, wang2020efficient}, but provide only rather weak guarantees on the quality of the estimates they return. In addition, even approximate approaches do not scale to billion edges networks. 

In this work we focus on the problem of counting motifs in temporal networks. Our goal is to obtain an efficiently computable approximation to the count of motifs of \emph{any} topology while providing rigorous guarantees on the quality of the result.

\subsubsection*{Our contributions}

This work provides the following contributions.
\begin{itemize}
    \item We present \algname, an algorithm to approximate the count of motifs in temporal networks, which provides rigorous (probabilistic) guarantees on the quality of the output. We present two variants of \algname, both based on a common approach that counts motifs within randomly sampled temporal windows. Both variants allow to analyze billion edges datasets providing sharp estimates. \algname~features several useful properties, including: 
    i) it has only one easy to interpret parameter, $c$, defining the length of the temporal windows for the samples; ii) it can approximate the count of \emph{any} motif topology; iii) it is easily parallelizable, with an almost linear speed-up with the available processors;
    \item We provide tight and efficiently computable bounds to the number of samples required by our algorithms to achieve (multiplicative) approximation error  $\le \varepsilon$ with probability $\ge 1 - \eta$, for given $\varepsilon >0$ and $\eta \in (0,1)$. Our bounds are obtained through the application of advanced concentration results (i.e., Bennett's inequality) for the sum of independent random variables. 
    \item We perform an extensive experimental evaluation on real datasets, including a dataset with more than 2.3 billion edges, never examined before. The results show that on large datasets our algorithm \algname\ significantly improves over the state-of-the-art sampling algorithms in terms of quality of the estimates while requiring a small amount of memory.
\end{itemize}


\section{Preliminaries}
\label{sec:prelim} 
In this section we introduce the basic definitions used throughout this work. We start by formally defining temporal networks.
\begin{Definition}
    \label{defn:temporal_graph}
    A \emph{temporal network} is a pair $T=(V,E)$ where, $V=\{v_1, \dots , v_n\}$ and $E=\{(x,y,t):x,y \in V, x \neq y, t \in \mathbb{R^{+}}\}$ with $|V|=n$ and $|E|=m$.
\end{Definition}

Given $(x,y, t) \in E$, we say that $t$ is the \emph{timestamp} of the edge $(x,y)$.
For simplicity in our presentation we assume the timestamps to be unique, which is without loss of generality since in practice our algorithms also handle non-unique timestamps. We also assume the edges to be sorted by increasing timestamps, that is $t_1 < \dots < t_m$. Given an interval or \emph{window} $[t_{B}, t_{E}] \subseteq \mathbb{R}$ we will denote $|t_{E}-t_{B}|$ as its \emph{length}.

We are interested in \emph{temporal motifs}\footnote{In static networks, the term \emph{graphlet}~\cite{yaverouglu2014revealing} is sometimes used, with \emph{motifs} denoting statistically significant graphlets. We use the term motif in accordance with previous work on temporal networks~\cite{paranjape_motifs_2017,liu_sampling_2019,wang2020efficient}.}, which are small, connected subgraphs whose edge timestamps satisfy some constraints. In particular, we consider the following definition introduced in~\cite{paranjape_motifs_2017}.

\begin{Definition}\label{defn:temporal-motif}
    A \emph{$k$-node $\ell$-edge temporal motif} $M$ is a pair $M = (\mathcal{K}, \sigma)$ where $\mathcal{K}=(V_\mathcal{K}, E_\mathcal{K})$ is a directed and weakly connected \emph{multigraph} where $V_{\mathcal K} =\{v_1, \dots , v_k\}$, $E_\mathcal{K}=\{(x,y):x,y \in V_\mathcal{K}, x \neq y\}$ s.t.\  $|V_\mathcal{\mathcal{K}}|=k$ and $|E_\mathcal{K}|=\ell$ and $\sigma$ is an ordering of $E_\mathcal{K}$.
\end{Definition}

Note that a $k$-node $\ell$-edge temporal motif $M = (\mathcal{K}, \sigma)$ is also identified by the sequence $(x_1,y_1), \dots,$ $ (x_{\ell},y_{\ell})$ of edges ordered according to $\sigma$. Given a $k$-node $\ell$-edge temporal motif $M$, $k$ and $\ell$ are determined by $V_\mathcal{K}$ and $E_\mathcal{K}$. We will therefore use the term \emph{temporal motif}, or simply \emph{motif}, when $k$ and $\ell$ are clear from context.

Given a temporal motif $M$, we are interested in counting how many times it appears within a time \emph{duration} of $\delta$, as captured by the following definition.

\begin{Definition}\label{defn:delta-instance}
    Given a temporal network $T=(V,E)$ and $\delta \in \mathbb{R^{+}}$, a \emph{time ordered} sequence $S=(x'_{1}, y'_{1}, t'_{1}), \dots, (x'_{\ell}, y'_{\ell}, t'_{\ell})$ of $\ell$ unique temporal edges from  $T$ is a $\delta$\textit{-instance} of the temporal motif $M=(x_1,y_1),\dots,(x_{\ell},y_{\ell})$  if:
    \begin{enumerate}
        \item there exists a bijection $f$ on the vertices such that $f(x'_{i}) = x_i$ and $f(y'_{i}) = y_i, \, i=1, \dots , \ell$; and
        \item the edges of $S$ occur within $\delta$ time, i.e., $t'_{\ell} - t'_{1} \leq \delta$.
    \end{enumerate} 
\end{Definition}

Note that in a $\delta$-instance of the temporal motif $M=(\mathcal{K},\sigma)$ the edge timestamps must be sorted according to the ordering $\sigma$.
See Figure~\ref{fig:basicdef} for an example.

\setlength{\textfloatsep}{12pt}
\begin{figure}[t]
    \centering
    \subfloat[]{
        \includegraphics[width=.3\linewidth]{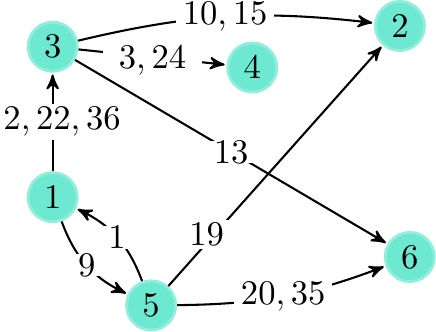}
        \label{subfig:tgraph}
    }\qquad
    \subfloat[]{
        \includegraphics[width = .3\linewidth]{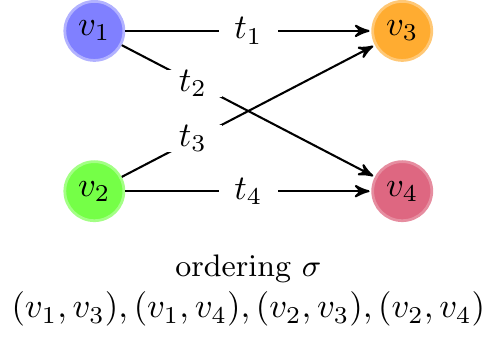}
        \label{subfig:tmotif}
    }\\
    \subfloat[]{
        \begin{tabular}{cccc}
            \includegraphics[width=.15\linewidth]{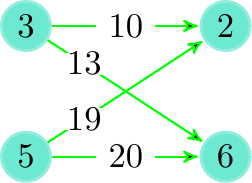} &  \includegraphics[width=.15\linewidth]{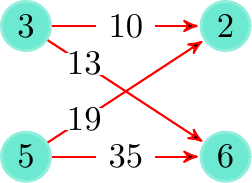} &
            \includegraphics[width=.15\linewidth]{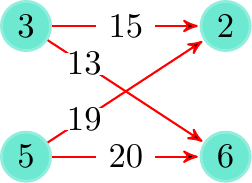} &
            \includegraphics[width=.15\linewidth]{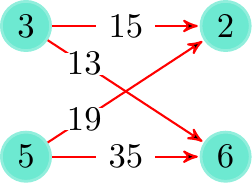}
        \end{tabular}
        \label{subfig:instances}
    }
    \caption{(\ref{subfig:tgraph}): example of temporal network $T$ with $n=6$ nodes and $m=13$ edges. (\ref{subfig:tmotif}): a temporal motif, known as Bi-Fan~\cite{liu_sampling_2019}. (\ref{subfig:instances}): sequences of edges of $T$ that map topologically on the Bi-Fan motif, i.e., in terms of static graph isomorphism. For $\delta= 10$ only the green sequence is a $\delta$-instance of the Bi-Fan, since the timestamps respect $\sigma$ and $t_{\ell}'-t_1' = 20-10 \le \delta$. The red sequences are not $\delta$-instances, since they do not respect such constraint or do not respect the order in $\sigma$.}
    \label{fig:basicdef}
\end{figure}

Let   $\mathcal{U}(M,\delta)=\{u: u$ is a $\delta$-instance of $M \}$ be the \textit{set of (all) $\delta$-instances} of the motif $M$ in $T$, denoted only with $\mathcal{U}$ when $M$ and $\delta$ are clear from the context. Given a $\delta$-instance $u \in \mathcal{U}(M,\delta)$, we denote the timestamps of its first and last edge with $t_1^u$ and $t_{\ell}^u$, respectively. The \emph{count} of $M$ is $C_M(\delta) = |\mathcal{U}(M,\delta)|$, denoted with $C_M$ when $\delta$ is clear from the context. The notation used throughout this work is summarized in Table \ref{tab:notationTable}. We are interested in solving the following problem.

\textbf{Motif counting problem.} Given a temporal network $T$, a temporal motif  $M = (\mathcal{K}, \sigma)$, and $\delta \in \mathbb{R^{+}}$, compute the count $C_M(\delta)$ of $\delta$-instances of motif $M$ in the temporal network $T$.

Solving the motif counting problem \emph{exactly} may be infeasible for large networks, since even  determining whether a temporal network contains a simple \emph{star motif} is NP-hard~\cite{liu_sampling_2019}. State-of-the-art exact techniques \cite{mackey_chronological_2018, paranjape_motifs_2017} require exponential time and memory in the number of edges of the temporal network, which renders them impractical for large temporal networks. We are therefore interested in obtaining efficiently computable approximations of motif counts, as follows.

\textbf{Motif approximation problem.} Given a temporal network $T$, a temporal motif  $M = (\mathcal{K}, \sigma)$, $\delta \in \mathbb{R^{+}}$, $\varepsilon \in \mathbb{R}^{+}_0, \eta \in (0,1)$ compute $C_{M}'$ such that $\mathbb{P}(|C_{M}' - C_{M}(\delta)| \geq \varepsilon C_{M}(\delta)) \leq \eta$, i.e., $C_{M}'$ is a \emph{relative $\varepsilon$-approximation} to $C_M(\delta)$ with probability at least $1-\eta$.

We call an algorithm that provides such guarantees an \emph{$(\varepsilon, \eta)$-approximation algorithm.}
\begin{table}[t]
    \centering
    \caption{Notation table.}
    \label{tab:notationTable}
        \begin{tabularx}{\textwidth}{cl}
            \toprule
            Symbol & Description\\
            \midrule
            $T=(V,E)$ & Temporal network\\
            $n,m$ & Number of nodes and number of temporal edges of $T$\\
            $t_i$ & Timestamp of the $i$-th edge $i=1,\dots,m$\\
            $M$ & Temporal motif\\
            $\ell, k$ & Edges and nodes of the temporal motif $M$\\
            $\delta$ & Duration of $\delta$-instances\\
            $\mathcal{U}$ & Set of $\delta$-instances of motif $M$ of $T$\\ 
            $C_M$ & Number of $\delta$-instances of $M$ in $T$\\
            $\varepsilon, \eta$ & Approximation and confidence parameters\\
            $s$ & Number of samples collected by \algname\ \\
            $\mathcal{S}_i$ & Set of $\delta$-instances sampled at $i$-th iteration from \algname \\
            $\Delta_{T,1}$ & Length of the window from which we sample $t_r$ in \algnameone\ \\ 
            $\Delta_{T,2}$ & Number of candidates for $t_r$ in \algnametwo \\ 
            $C_M'$ & Estimate in output by \algname\\
            $X_i$ & $i$-th random variable (i.e., estimate) obtained by \algname\\
            $X_u^i, \tilde{X}_u^i$ & \makecell[l]{Indicator random variables denoting if $u\in \mathcal{U}$ has been sampled at the $i$-th iteration, $i\in [1,s]$\\ in \algnameone\ and \algnametwo\ respectively}\\
            $p_u, \tilde{p}_u$ & \makecell[l]{Probability of selecting a sample that contains $u\in \mathcal{U}$ in \algnameone\ and \algnametwo\ respectively.}\\
            \bottomrule
        \end{tabularx}
\end{table}
\subsection{Related Work}
\label{sec:relatedwork}
Various definitions of temporal networks and motifs have been proposed in the literature; we refer the interested reader to~\cite{holme_modern_2015,holme_temporal_2019, liu2020temporal}. Here we focus on those works that adopted the same definitions used in this work.

The definition of temporal motif we adopt was first proposed by Paranjape et al.~\cite{paranjape_motifs_2017}, which provided efficient exact algorithms to solve the motif counting problem for specific motifs. 
Such algorithms are efficient only for specific motif topologies and do not scale to very large datasets.
An algorithm for the counting problem on general motifs has been introduced by Mackey et al.\ \cite{mackey_chronological_2018}. Their algorithm is the first exact technique allowing the user to enumerate all $\delta$-instances $u \in \mathcal{U}$ without any constraint on the motif's topology. The major back-draw of such algorithm is that it may be impractical even for moderately-sized networks, due to its exponential time complexity and memory requirements.

Liu et al.~\cite{liu_sampling_2019} proposed the first sampling algorithm for the temporal motif counting problem. The main strategy of \cite{liu_sampling_2019} is to partition the time interval containing all the edges of the network into \emph{non overlapping} and contiguous windows of length $c \delta$ (i.e., a grid-like partition), for some $c>1$. The partition is then randomly shifted (i.e., the starting point of the first window may not coincide with the smallest timestamp of the network). The edges in each partition constitute the candidate samples to be analyzed using an exact algorithm. 
An \emph{importance sampling} scheme is used to sample (approximately) $r$ windows among the candidates, with a window being selected with probability proportional to the fraction of edges it contains.
The estimate for each sampled window is obtained by weighting each $\delta$-instance in the window, and the estimates are averaged across windows. This procedure is repeated $b$ times to reduce the variance of the estimate. While interesting, this partition-based approach prevents such algorithm to provide $(\varepsilon, \eta)$-guarantees (see Section \ref{subsec:generalSchema}).

Recently Wang et al.~\cite{wang2020efficient} proposed an $(\varepsilon, \eta)$-approximation algorithm for the motif counting problem. Their approach selects each edge in $T$ with a user-provided probability $p$. Then for each selected  edge $e=(x,y,t)$, the algorithm collects the edges with timestamps in the edge-centered window $[t-\delta, t+\delta]$, of length 2$\delta$, computes on these edges all the $\delta$-instances $u \in \mathcal{U}$ containing $e$, weights the instances, and combines the weights to obtain the final estimate. From the theoretical point of view, the main drawback of this approach is that in order to achieve the desired guarantees one has to set $p \ge 1/(1+\eta \varepsilon^2)$, resulting in high values of $p$ (i.e., almost all edges are selected) for reasonable values of $\eta$ and $\varepsilon$ (e.g., $p>0.97$ for $\eta = 0.1$ and  $\varepsilon = 0.5$).
In addition, such approach is impractical on very large datasets, mainly due its huge memory requirements, and does not provide accurate estimates (see  Section~\ref{sec:experimentaleval}). 

\section{\algname: Approximating Temporal Motif Counts with Uniform Sampling}
\label{sec:algs}

We now describe and analyze our algorithm \algname\ (ap\underline{PR}oximating t\underline{E}mporal motif\underline{S} coun\underline{T}s with unif\underline{O}rm sampling) for the motif counting problem. We start by describing, in Sec.\ \ref{subsec:generalSchema}, the common strategy underlying our algorithms.
We then briefly highlight the differences between \algname\ and the existing sampling algorithms for the counting problem. In Sec.\ \ref{sec:RSRN} and Sec.\ \ref{sec:RSRT} we present and analyze two variants of the common strategy introduced in Sec.\ \ref{subsec:generalSchema}. We conclude with the analysis of \algname's complexity in Sec.\ \ref{subsec:complexity_analysis}.

\subsection{\algname: General Approach}
\label{subsec:generalSchema}
The general strategy of \algname~is presented in Algorithm~\ref{alg:generalschema}.
Given a network $T$, \algname\ collects $s$ samples, where each sample is obtained by gathering all edges $e \in E$ with timestamp in a \emph{small random} window $[t_r, t_{r} + c \delta]$ of length $c \delta$, using \texttt{SampleWindowStart($E$)} (Lines~\ref{algline:line_t_r}-\ref{algline:line_T_i}) to select $t_r$ (i.e., the starting time point of the window) through a uniform random sampling approach. 
For each sample, an exact algorithm is then used to extract all the $\delta$-instances of motif $M$ in such sample (Line \ref{algline:exactdeltainsts}). The weight of each $\delta$-instance extracted is computed  with the call \texttt{ComputeWeight($u$)} (Line~\ref{algline:line_wu}), and the sum of all such weights constitutes the estimate provided by the sample (Line~\ref{algline:ith-estimate}). The weights account for the probability of sampling the $\delta$-instances in the sample, making the final estimate \emph{unbiased}. The final estimate produced in output is the average of the samples' estimates (Lines~\ref{algline:finalestimate}-\ref{algline:output}). Note that the \texttt{for} cycle of Line~\ref{line_for} is trivially parallelizable.  The two variants of \algname\  we will present differ in the way i) \texttt{SampleWindowStart($E$)} and ii) \texttt{ComputeWeight($u$)} are defined.
\begin{algorithm2e}[t]
    \LinesNumbered 
    \KwIn{Temporal Network ${T}=(V,E)$, Motif $M$, Motif Duration $\delta > 0$, $s> 0$, $c>1$}
    \KwOut{Estimate $C_{M}'$ of $C_{M}$}
    \For{$i \gets 1$ \text{to} $s$\label{line_for}}
    {
        $t_r \gets$ \texttt{SampleWindowStart}($E$)\label{algline:line_t_r}\;
        $E_i \gets \{(x,y,t) \in E : (t \ge t_r) \land (t \le t_r + c \delta)  \}$\label{algline:edges}\;
        $V_i \gets \{x : ((x,y,t) \in E_i) \lor ((y,x,t) \in E_i) \}$\label{algline:nodes}\;
        $T_i \gets (V_i, E_i)$\label{algline:line_T_i}; $X_i \gets 0$\;
        $\mathcal{S}_i \gets$ $\{u: u \text{ is } \delta\text{-instance of } M \text{ in } T_i \}$\label{algline:exactdeltainsts}\;
        \ForEach{$u \in \mathcal{S}_i$}
        {
            $w(u) \gets$ \texttt{ComputeWeight}($u$)\label{algline:line_wu}\;
            $X_i \gets X_i + w(u)$\label{algline:ith-estimate}\;
        }
    }
    $C_M' \gets \frac{1}{s} \sum_{i=1}^{s} X_i$\label{algline:finalestimate}\;
    \Return $C_M'$\label{algline:output}\;
    \caption{\algname}\label{alg:generalschema}
\end{algorithm2e}

Differently from~\cite{liu_sampling_2019}, our algorithm \algname\ does not partition the edges of $T$ in non overlapping windows, and relies instead on \emph{uniform sampling}.
We recall (see Section~\ref{sec:relatedwork}) that in \cite{liu_sampling_2019}, after computing all the non overlapping intervals defining the candidate samples, there may be several $\delta$-instances $u\in \mathcal{U}$ that cannot be sampled (i.e., all $\delta$-instances having $t_1^u$ in window $j$ and $t_{\ell}^u$ in window $j+1$). This significantly differs from \algname, which instead samples at each iteration a small random window from $[t_1,t_m]$ without restricting the candidate windows, allowing to sample \emph{any} $\delta$-instance $u\in \mathcal{U}(M,\delta)$ at each iteration. 
This enables us to provide stronger guarantees on the quality of the output, since each $\delta$-instance has a non-zero probability of being sampled at each step, leading to  $(\varepsilon, \eta)$-approximation guarantees.

Differently from the work of Wang et al.~\cite{wang2020efficient} \algname\ samples temporal windows of length $c\delta$ and does not follow the edge-centric approach (i.e., sampling temporal edges with a user-provided probability $p$) of~\cite{wang2020efficient}. In addition, the approach of~\cite{wang2020efficient} collects edges in temporal-windows of length $2\delta$ (see Section~\ref{sec:relatedwork}), while in \algname\ the window size is controlled by the parameter $c$, which, when fixed to be $<2$, leads to much more scalability than~\cite{wang2020efficient} while requiring less memory (see Section~\ref{sec:experimentaleval}). 


\subsection{\algnameone: Sampling among All Windows}
\label{sec:RSRN}

In this section we present and analyze \algnameone, our first $(\varepsilon, \eta)$-approximation algorithm obtained by specifying i) how the starting point $t_r$ of the temporal window defining a sample $T_i$  is chosen (function \texttt{SampleWindowStart$(E)$} in Line~\ref{algline:line_t_r}) and ii) how the weight $w(u)$ of a $\delta$-instance $u$ in a sample is computed (\texttt{ComputeWeight$(u)$}, Line~\ref{algline:line_wu}).

The starting point $t_r$ of sample $T_i$ is sampled uniformly at random in the interval $[t_{\ell} - c\delta, t_{m-\ell}] \subseteq \mathbb{R}$, where we recall $\ell = |E_\mathcal{K}|$ and $m = |E|$.  
Regarding the choice of the weight $w(u)$ for each instances $u \in \mathcal{S}_i$,  \algnameone\ considers $w(u) = 1 / p_u$, with $p_u$ being the probability of $u$ to be in $\mathcal{S}_i$, that is  $p_u = r_u / \Delta_{T,1}$, where $\Delta_{T,1} = t_{m-\ell}-t_{\ell} + c \delta$  is the 
total length of the interval from which $t_r$ is sampled (recall we choose $t_r$ from $[t_{\ell} - c\delta, t_{m-\ell}]$),
and $r_u = c\delta - (t^{u}_{\ell} - t^{u}_{1}) $ is the length of the interval in which $t_r$ must be chosen for $u$ to be in $\mathcal{S}_i, i=1,\dots,s$.

We now present the theoretical guarantees of \algnameone\ and give efficiently computable bounds for the sample size $s$ needed for the $(\varepsilon, \eta)$-approximation to hold. 
Recall the definition of $\mathcal{U}(M,\delta)$, which is the set of $\delta$-instances of $M$ in $T$. Let $u$ be an arbitrary $\delta$-instance of motif $M$, and let $T_i$ be an arbitrary sample obtained by  \algnameone~at its $i$-th iteration. We define the following set of indicator random variables, for $u \in \mathcal{U}$ and for $i=1, \dots,s$: $X_{u}^{i}=1 \text{ if $u \in \mathcal{S}_i$, 0 otherwise.}$
Each variable $X_u^i, i=1,\dots,s\, , u \in \mathcal{U}$ is a Bernoulli random variable with $\mathbb{P}(X_{u}^{i} = 1) = \mathbb{P}(u \in \mathcal{S}_i) =
\frac{r_{u}}{\Delta_{T,1}} = p_u.$
Therefore, for each variable $X_u^i$, it holds $\mathbb{E}[X_u^i]= p_u$. 
Thus, for each $i=1,\dots,s$, iteration $i$  provides an estimate $X_i$ of $C_M$, which is the random variable: $X_i = \sum_{u \in \mathcal{U}} \frac{1}{p_u}X_{u}^i.$
We therefore have  the following result.
\begin{Lemma}\label{lemma:unbiased_estimate}
    For each  $i=1,\dots, s$, $X_i$ and $ C_M' = \frac{1}{s} \sum_{i=1}^s X_i$ are unbiased estimators for $C_M$, that is $\mathbb{E}[X_i] =\mathbb{E}[C_M']
    =C_M$.
\end{Lemma}
\begin{proof}
    Considering the linearity of expectation and the definition of the variables $X_i$ and $X_u^i, i=1,\dots,s, u \in \mathcal{U}$, we have:
    \begin{displaymath}\label{eq:proof_unbiased_RSRN}
    \mathbb{E}[X_i] \buildrel{}\over{=} \sum_{u \in \mathcal{U}} \frac{1}{p_u} \mathbb{E}\left[X_{u}^{i}\right] =
    \sum_{u \in \mathcal{U}} \frac{1}{p_u} p_u = C_M.
    \end{displaymath}
    Combining with the linearity of expectation to $C_M'$ the statement follows.
\end{proof}

The following result provides a bound on the variance of the estimate of $C_M$ provided by $C_M'$.
\begin{Lemma}\label{lemma:variance}
    For \algnameone~it holds
    \begin{displaymath}
    \mathrm{Var}\left(C_M'\right) = \mathrm{Var}\left(\frac{1}{s} \sum_{i=1}^{s} X_i\right) \leq \frac{C_M^2}{s}\left( \frac{\Delta_{T,1}}{(c-1)  \delta} -1\right).
    \end{displaymath}
\end{Lemma}
\begin{proof}
    We start by observing that,
    $
    \mathrm{Var}\left(\frac{1}{s} \sum_{i=1}^{s} X_i\right) = \frac{1}{s^2} \sum_{i=1}^{s} \mathrm{Var}(X_i)
    $
    by the mutual independence of the variables $X_1, \dots, X_s$.
    We observe that $\mathrm{Var}(X_i)= \mathbb{E}[X_i^2] - \mathbb{E}[X_i]^2, i=1,\dots,s$ and we recall that $\mathbb{E}[X_i] = C_M$ by Lemma \ref{lemma:unbiased_estimate}, 
    thus we need to bound $\mathbb{E}[X_i^2]$,
    \begin{equation}\label{eq:bound_squared_Xi}
    \begin{split}
    \mathbb{E}[X_i^2] = 
    \sum_{u_1 \in \mathcal{U}} \sum_{u_2 \in \mathcal{U}} \frac{1}{p_{u_1} p_{u_2}}\mathbb{E}[X_{u_1}^{i}X_{u_2}^{i}]
    \buildrel{(1.)}\over{\le} \sum_{u_1 \in \mathcal{U}} \sum_{u_2 \in \mathcal{U}} \frac{1}{p_{u_2}} 
    \buildrel{(2.)}\over{\leq} \sum_{u_1 \in \mathcal{U}} \sum_{u_2 \in \mathcal{U}} \frac{\Delta_{T,1}}{(c-1) \delta} = \frac{C_M^2 \Delta_{T,1}}{(c-1) \delta}
    \end{split}
    \end{equation}
    where (1.) follows from $\mathbb{E}[XY] \le \mathbb{E}[X]$  and (2.) from the definition of $p_{u_2}$.
    Based on the above, we can bound the variance of the variables $X_i, i=1,\dots,s$, as follows,
    $
    \mathrm{Var}(X_i) = \mathbb{E}[X_i^2] - \mathbb{E}[X_i]^2 \le C_M^2 \left(\frac{ \Delta_{T,1}}{(c-1) \delta} -1  \right).
    $
    Note that such bound does not depend on the index $i=1,\dots,s$, thus substituting in the original summation we obtain,
    $ \mathrm{Var}\left(\frac{1}{s} \sum_{i=1}^{s} X_i\right) \leq  \frac{C_M^2}{s} \left(\frac{ \Delta_{T,1}}{(c-1) \delta} -1  \right).$
\end{proof}
We now present a first efficiently computable bound on the number of samples $s$ required to have that $C_M'$ is a relative $\varepsilon$-approximation of $C_M$ with probability $\ge 1 - \eta$. 
Such bound is based on the application of Hoeffding's inequality (see \cite{mitzenmacher2017probability}), 
an advanced but commonly used technique in the analysis of probabilistic algorithms. Let us first introduce Hoeffding's inequality,
\begin{theorem}[{\cite{hoeffding_probability_1963}}]\label{theo:hoeffding}
    Let $X_1, \dots, X_s$ be independent random variables such that for all $1 \le i \le s$, $\mathbb{E}[X_i]=\mu$ and $\mathbb{P}(a \le X_i \le b) = 1.$ Then
    \begin{displaymath}
    \mathbb{P} \left( \left|\frac{1}{s}\sum_{i=1}^s X_i - \mu \right|\ge t \right) \le 2 \exp\left(-\frac{2st^2}{(b-a)^2}\right)
    \end{displaymath}
\end{theorem}
then the bound on the sample size we derive based on the above is as follows, 
\begin{theorem} \label{theo:bound_s_sloppy_RSRN}
    Given $\varepsilon \in \mathbb{R}^{+}, \eta \in (0,1)$ let $X_1, \dots, X_s$ be the random variables associated with the counts computed at iterations $1,\dots,s$, respectively, of \algnameone. If $s \geq  \frac{\Delta_{T,1}^2}{2(c-1)^2 \delta^2\varepsilon^2} \ln \left( \frac{2}{\eta} \right) $, then it holds
    \begin{displaymath}
    \mathbb{P}\left( \left|\frac{1}{s}\sum_{i=1}^s X_i - C_{M}\right| \geq \varepsilon C_{M} \right) \leq \eta
    \end{displaymath}
    that is, \algnameone~is an $(\varepsilon, \eta)$-approximation algorithm.
\end{theorem}
\begin{proof}
    
    Let $Pr = \mathbb{P}(|\frac{1}{s}\sum_{i=1}^s X_i - C_{M}|$ $\geq \varepsilon C_{M} )$, we want to show that for  s as in statement it holds $Pr \le \eta$.  Observe that the variables $X_1,\dots,X_s$, are mutually independent. And, by Lemma 
    \ref{lemma:unbiased_estimate}, $\mathbb{E}[X_i] = C_M = \mu, i=1,\dots,s$. Then for $X_i, i=1,\dots,s$, the following holds:
    \begin{equation}\label{eq:bound_X_i_RSRN}
    X_i = \sum_{u \in \mathcal{U}} \frac{1}{p_u}X_{u}^i 
    \buildrel{(1.)}\over{\le} \sum_{u \in \mathcal{U}} {\frac{\Delta_{T,1}}{(c-1) \delta}} \buildrel{}\over{\le} \frac{C_M \Delta_{T,1}}{(c-1) \delta}
    \end{equation}
    where $(1.)$ follows from  the definition of $p_u$, by considering each of the variables $X_u^i =1$, and noting that for each $u \in \mathcal{U}$ it holds $r_u = c \delta - (t^{u}_{\ell} - t^{u}_{1}) \ge (c-1) \delta$. 
    Thus w.p.\ 1, $a = 0 \le X_i \le {C_M \Delta_{T,1}/{((c-1) \delta)}} = b, i=1,\dots,s$. Now  let $t = \varepsilon C_M$. Then,
    $
    (b-a)^2 = \left(\frac{ \Delta_{T,1}}{(c-1) \delta}C_M -0\right)^ 2 = \frac{ \Delta_{T,1}^2 C_M^2}{(c-1)^2 \delta^2} 
    $
    and applying Theorem \ref{theo:hoeffding} to bound $Pr$ we obtain,
    \begin{displaymath}
    Pr \le 2 \exp\left(-\frac{2s\varepsilon^2C_M^2}{\frac{ \Delta_{T,1}^2 C_M^2}{(c-1)^2 \delta^2}}\right) \le \eta 
    \end{displaymath}
    by the choice of $s$ as in statement.
\end{proof}

We now show that by using Bennett's inequality (see \cite{boucheron2013concentration}), a more advanced result on the concentration of the sum of independent random variables, we can derive a bound that is much tighter than the above, while still being efficiently computable. We first state Bennett's inequality,
\begin{theorem}[{\cite{bennett_probability_1962}}]\label{theo:bennet} 
    For a collection $X_1,\dots,X_s$ of independent random variables satisfying $X_i \le M_i$, $\mathbb{E}[X_i]= \mu_i$ and $\mathbb{E}[(X_i-\mu_i)^2]= \sigma_i^2$ for $i=1,\dots,s$ and for any $t\ge 0$, the following holds
    \begin{displaymath}
    \mathbb{P} \left( \left|\frac{1}{s}\sum_{i=1}^s X_i - \frac{1}{s}\sum_{i=1}^s \mu_i \right|\ge t \right) \le 2\exp\left(-s\frac{v}{B^2} h\left(\frac{tB}{v}\right) \right)
    \end{displaymath}
    where $h(x) = (1+x) \ln(1+x)-x, B= \max_i M_i - \mu_i$ and $v=\frac{1}{s}\sum_{i=1}^s \sigma_i^2$.
\end{theorem}
One of the difficulties in applying such result to our scenario is that we know only an upper bound to the variance of the random variables $X_i, i=1, \dots,s$. We discuss why Bennett's inequality can be applied in such situation in Supplement \ref{app:bennett_note}. 
We now state our main result.
\begin{theorem}
    \label{theo:bound_s_RSRN}
    Given $\varepsilon \in \mathbb{R}^{+}, \eta \in (0,1)$ let $X_1, \dots, X_s$ be the random variables associated with the counts computed at iterations $1,\dots,s$, respectively, of \algnameone. If $s \geq  \left(\frac{ \Delta_{T,1}}{(c-1) \delta} - 1 \right)\frac{1}{(1+\varepsilon) \ln(1+\varepsilon) - \varepsilon} \ln \left( \frac{2}{\eta} \right)$, then it holds
    \begin{displaymath}
    \mathbb{P}\left( \left|\frac{1}{s}\sum_{i=1}^s X_i - C_{M}\right| \geq \varepsilon C_{M} \right) \leq \eta
    \end{displaymath}
    that is, \algnameone~is an $(\varepsilon, \eta)$-approximation algorithm.
\end{theorem}
\begin{proof}
    Let $Pr = \mathbb{P}\left(|\frac{1}{s}\sum_{i=1}^s X_i - C_{M}| \geq \varepsilon C_{M} \right)$, we want to show that for $s$ as in the statement it holds $Pr \le \eta$. 
    
    In order to apply Bennett's bound, note that
    $\frac{1}{s}\sum_{i=1}^s \mathbb{E}[X_i] = \frac{1}{s}\sum_{i=1}^s \mu_i= C_M$ by Lemma \ref{lemma:unbiased_estimate}
    . Moreover $B =\max_i M_i - \mu_i  = \frac{ \Delta_{T,1}}{(c-1) \delta} C_M - C_M = C_M\left(\frac{ \Delta_{T,1}}{(c-1) \delta} -1 \right)$ and in addition 
    $\mathbb{E}[(X_i-\mu_i)^2]= \sigma_i^2 \le C_M^2 \left(\frac{ \Delta_{T,1}}{(c-1) \delta} -1 \right) = \hat\sigma_i^2, i=1,\dots,s$ (see Eq.\ \eqref{eq:bound_squared_Xi} and Eq.\ \eqref{eq:bound_X_i_RSRN}), thus $v=\frac{1}{s}\sum_{i=1}^s \sigma_i^2 \le \frac{1}{s}\sum_{i=1}^s \hat\sigma_i^2 = \frac{1}{s}\sum_{i=1}^s  C_M^2 \left(\frac{ \Delta_{T,1}}{(c-1) \delta} -1\right) =  C_M^2 \left(\frac{ \Delta_{T,1}}{(c-1) \delta} -1\right) = \hat v$. Let $t = \varepsilon C_M$, then it holds
    \item \begin{displaymath}\label{eq:boundinner_RSRN}
    \begin{split}
    \frac{tB}{\hat v} = \varepsilon \qquad \text{ and }\qquad \frac{\hat v}{B^2} = \frac{1}{\left( \frac{ \Delta_{T,1}}{(c-1) \delta}-1 \right)}
    \end{split}
    \end{displaymath}
    applying Theorem~\ref{theo:bennet}, with the upper bound $\hat{v}$ to $v$ we obtain,
    \begin{displaymath} 
    \begin{split}
    Pr \buildrel{}\over{\le}2 \exp \left(-s\frac{ \hat v}{B^2} h\left(\frac{tB}{ \hat v}\right) \right) 
    \le \eta
    \end{split}
    \end{displaymath}
    by substituting the quantities above and by the choice of $s$ as in statement, which concludes the proof. 
\end{proof}

Note that the bound in Theo.\ \ref{theo:bound_s_RSRN} is significantly better than the one in Theo.\ \ref{theo:bound_s_sloppy_RSRN}, since the former has a quadratic dependence from $\Delta_{T,1}$ while the latter enjoys a linear dependence from $\Delta_{T,1}$. Furthermore,  differently from the bounds in \cite{wang2020efficient}, our bounds depend on characteristic quantities of the datasets ($\Delta_{T,1}$) and of the algorithm's input $(c, \delta)$. Therefore we expect our bounds to be more informative, and also possibly tighter than the bounds of~\cite{wang2020efficient} (with the improvement due, at least in part, to the use of Bennett's inequality, while~\cite{wang2020efficient} leverages on Chebyshev's inequality, which usually provides looser bounds \cite{mitzenmacher2017probability}). 
We conclude by giving an intuition on the bound in Theorem \ref{theo:bound_s_RSRN}, observe in fact that the bound is directly proportional to the quantity $\Delta_{T,1}/((c-1)\delta)$, such quantity (approximatively) corresponds to the fraction between the timespan of the temporal network $T$ and the length of each window we draw, thus such bound suggests (in a very intuitive way) that as long as the windows we draw have a larger length (i.e., $\delta$ or $c$ become larger) then less samples are needed for \algnameone\ to provide concentrate estimates within the desired theoretical guarantees.  Note that windows with a larger duration will require a higher running time to be processed, as we will discuss in Section \ref{subsec:complexity_analysis}.


\subsection{\algnametwo: Sampling the Start from Edges}\label{sec:RSRT}
In this section we present and analyse our second $(\varepsilon,\eta)$ approximation algorithm, \algnametwo, obtained with a different variant of the general strategy of Algorithm~\ref{alg:generalschema}.

\algnametwo\ selects the starting point $t_r$ (Line \ref{algline:line_t_r} in Alg.\ \ref{alg:generalschema}) of sample $T_i$ only from the timestamps of edges of $T$. In particular, $t_r$ is chosen
uniformly at random in $\{t_1, \dots, t_{last}\} \subseteq \{t_1,\dots,t_m\}$, where $t_{last} = \min\{t: (x,y,t) \in E \land t \ge t_m - c \delta\}$. 
When the edges of $T$ are far from each other or have a skewed distribution in time, or occur mostly in some well-spaced subsets $\{[t_a,t_b] \subseteq [t_1,t_m]: t_a \le t_b \}$, 
we expect \algnametwo~to collect samples with more $\delta$-instances than \algnameone~(which samples $t_r$ uniformly at random almost from the entire time interval $[t_1,t_m]$ but without restricting to existing edges).
Similarly to \algnameone, the weight $w(u)$ (Line \ref{algline:line_wu} of Alg.\ \ref{alg:generalschema}) is computed by $w(u) = 1 / \tilde{p}_u$, where $\tilde{p}_u$ is the probability that $u \in \mathcal{U}$ is in the set $\mathcal{S}_i$. 
Due to the (random) choice of $t_r$, $\tilde{p}_u = \tilde{r}_{u}/\Delta_{T,2}$, where $\Delta_{T,2} = |\{t:(x,y,t) \in E \land t \in [t_1, t_{last}] \}|$  is the number of possible choices for $t_r$ (if $t_{last} = t_m$ then $\Delta_{T,2} = m$) and  $\tilde{r}_u = |\{t: (x,y,t) \in E \land t \in [\max\{t_1, t_{\ell}^{u} - c\delta\},\min\{t_{last},t_{1}^{u}\}] \}|$ is the number of choices for $t_r$ such that $u \in \mathcal{U}$ is in $\mathcal{S}_i, i =1,\dots,s$.

Similarly as done for \algnameone\, we prove that \algnametwo\ provides an unbiased estimate
of $C_M$, with bounded variance (the proof is in Supplement \ref{subsec:theoguaran_appendix}).
By applying Bennett's inequality, we derive the following.

\begin{theorem} \label{theo:bound_s}
    Given $\varepsilon \in \mathbb{R}^{+}, \eta \in (0,1)$ let $X_1, \dots, X_s$ be the random variables associated with the counts computed at iterations $1,\dots,s$, respectively, of \algnametwo. If $s \geq  \frac{(\Delta_{T,2} -1)}{(1+\varepsilon) \ln(1+\varepsilon) - \varepsilon} \ln \left( \frac{2}{\eta} \right) $, then it holds
    \begin{displaymath}
    \mathbb{P}\left( \left|\frac{1}{s}\sum_{i=1}^s X_i - C_{M}\right| \geq \varepsilon C_{M} \right) \leq \eta
    \end{displaymath}
    that is, \algnametwo~is an $(\varepsilon, \eta)$-approximation algorithm.
\end{theorem}

As for \algnameone, by using Bennett's inequality we obtain a linear dependence between the sample size $s$ and $\Delta_{T,2}$, while commonly used techniques (e.g., Hoeffding's inequality) lead to a quadratic dependence.


\subsection{\algname: Complexity Analysis}\label{subsec:complexity_analysis}In this section we analyse the time complexity of \algname's versions introduced in the previous sections.
 
Note that our algorithms employ an exact enumerator as subroutine. For the sake of the analysis we consider the complexity when the subroutine is the algorithm by Mackey et al.\ 
\cite{mackey_chronological_2018}, which we used in our implementation.
Let us first start with a definition, given a temporal network $T=(V,E)$, we denote with $\dim(T)$ the set $\{t_i | (x,y,t_i)\in E, i=1,\dots,m\}$. 
The algorithm in \cite{mackey_chronological_2018} has a worst case complexity $O(m \hat{\kappa}^{(\ell-1)})$ when executed on a motif with $\ell$ edges and a temporal network with $m$ temporal edges, where $\hat{\kappa}$ is the maximum number of edges within a window of length $\delta$, i.e.,   $\hat{\kappa} = \max\{|S(t)|: S(t) = \{(x,y,\bar{t}) \in E : \bar{t} \in [t,t+\delta]\}, t \in \dim(T)\}$. Observe that such complexity is  \emph{exponential} in the number of edges $\ell$ of the temporal motif. Obviously our algorithms benefit from any improvement to the state-of-the-art exact algorithms.

\textbf{Complexity of \algnameone.} The worst case complexity is $O(s\hat{m} \hat{\kappa}^{(\ell-1)} + sC_M^*)$ when executed sequentially, where $\hat{m}$ 
is the maximum number of edges in a window of length $c\delta$, 
i.e., $\hat{m} = \max \{|S(t)|: S(t) = \{(x,y,\bar{t}) \in E : \bar{t} \in [t,t+c\delta]\} , t \in \dim(T)\}$, and $C_M^*$ is the maximum number of $\delta$-instances contained in any window of length $c \delta$. This corresponds to the case where \algnameone~collects samples with many edges for which many $\delta$-instances occur. The complexity becomes $O(\frac{s}{\tau}\hat{m} \hat{\kappa}^{(\ell-1)} + \frac{s}{\tau} C_M^*)$ when \algnameone\ is executed in a parallel environment with $\tau$ threads (parallelizing the \texttt{for} cycle in Alg.\ \ref{alg:generalschema}).

\textbf{Complexity of \algnametwo.} Using the notation defined above, the worst-case complexity of a naive implementation of \algnametwo~is $O(s\hat{m} \hat{\kappa}^{(\ell-1)} + s\hat{m} C_M^*)$. As for \algnameone, the first term comes from the complexity of the exact algorithm for computing $\mathcal{S}_i, i=1,\dots,s$, whereas the second term is the worst case complexity cost of computing the weights for each sample. 
The additional $O(\hat{m})$ complexity of the second term arises from the computation of $\tilde{r}_u$ for each $u \in \mathcal{S}_i, i=1,\dots,s$. The complexity of this computation can be reduced to $O(\log(m))$ by applying binary search to the edges of $T$. With such an approach, we obtained a final complexity $O(s\hat{m} \hat{\kappa}^{(\ell-1)} + s \log(m) C_M^*)$.
The complexity reduces to $O(\frac{s}{\tau}\hat{m} \hat{\kappa}^{(\ell-1)} + \frac{s}{\tau} \log(m)  C_M^*)$ when $\tau$ threads are available for a parallel execution.

\section{Experimental evaluation}
\label{sec:experimentaleval}

In this section we present the results of our extensive experimental evaluation on large scale datasets. To the best of our knowledge we consider the largest dataset, in terms of number of temporal edges, ever used for the motif counting problem with more than 2.3 billion temporal edges.  
After describing the experimental setup and implementation (Sec.\ \ref{subsec:expsetup}), we first compare the quality of the estimated counts provided by \algname\ with the estimates from state-of-the-art sampling algorithms from Liu et al.~\cite{liu_sampling_2019}, and Wang et al.~\cite{wang2020efficient} (Sec.\ \ref{subsec:approxfact}). Then we compare the memory requirements of the algorithms (Sec.\ \ref{subsec:memory}), which may be a limiting factor for the analysis of very large datasets. Additional experiments on \algname's running time comparison with the exact algorithm by Mackey et al.~\cite{mackey_chronological_2018}, and with the current state of-the art sampling algorithms are in Supplement \ref{app:runningtime}. While results showing the linear scalability of     \algname's parallel implementation with the available number of processors can be found in Supplement \ref{app:scalability_suppl}.

\subsection{Experimental Setup and Implementation}\label{subsec:expsetup}
The characteristics of the datasets we considered are in Table~\ref{tab:datasets}. EquinixChicago \cite{noauthor_caida_nodate} is a \emph{bipartite} temporal network that we built, its description can be found in Supplement \ref{app:EquinixDataset}. 
Descriptions of the other datasets are in~\cite{paranjape_motifs_2017, hessel2016science, liu_sampling_2019}. 

We implemented our algorithms in C++14, using 
the algorithm 
by Mackey et al.~\cite{mackey_chronological_2018} as subroutine for 
the exact solution of the counting problem on samples.
We ran all experiments on a 64-core Intel Xeon E5-2698 with 512GB of RAM machine, with Ubuntu 14.04. The code we developed, and the dataset we build are entirely available at \url{https://github.com/VandinLab/PRESTO/}, links to additional resources (datasets and other implementations) are in Supplement \ref{app:links}.
We tested all the algorithms on the motifs from Figure~\ref{fig:motifs}, that are a temporal version of graphlets from \cite{prvzulj2004modeling}, often used in the analysis of static networks. 
We considered motifs with at most $\ell=4$ edges since the implementation by Wang et al.~\cite{wang2020efficient} does not allow for motifs with a higher $\ell$ in input. Since the EC dataset is a bipartite network, it can contain all but the motifs marked with a rectangle in Fig.\ \ref{fig:motifs}. Therefore, for such dataset we did not consider the motifs in the rectangles of Fig.\ \ref{fig:motifs}.
We compared our algorithms \algnameone\ and \algnametwo, collectively denoted as \algname, with 2 baselines: the sampling algorithm by Liu et al.~\cite{liu_sampling_2019}, denoted by \texttt{LS} and the sampling algorithm by Wang et al.~\cite{wang2020efficient}, denoted by \texttt{ES}.
\begin{table}[t]
    \centering
    \caption{Datasets used in our experimental evaluation. We report: number of nodes $n$; number of edges $m$; \emph{precision} of the timestamps; \emph{timespan} of the network.}
    \label{tab:datasets}
        \begin{tabular}{ccccl}
            \toprule
            Name& $n$&$m$& Precision & Timespan\\
            \midrule
            Stackoverflow (SO) & 2.58M & 47.9M & sec & 2774 (days)\\
            Bitcoin (BI) & 48.1M & 113M &sec & 2585 (days)\\
            Reddit (RE) & 8.40M & 636M & sec &3687 (days)\\
            EquinixChicago (EC) & 11.16M & 2.32B & $\mu$-sec &62.0 (mins)\\
            \bottomrule
        \end{tabular}
\end{table}

\begin{figure}[t]
    \centering
    \includegraphics[width=0.45\linewidth]{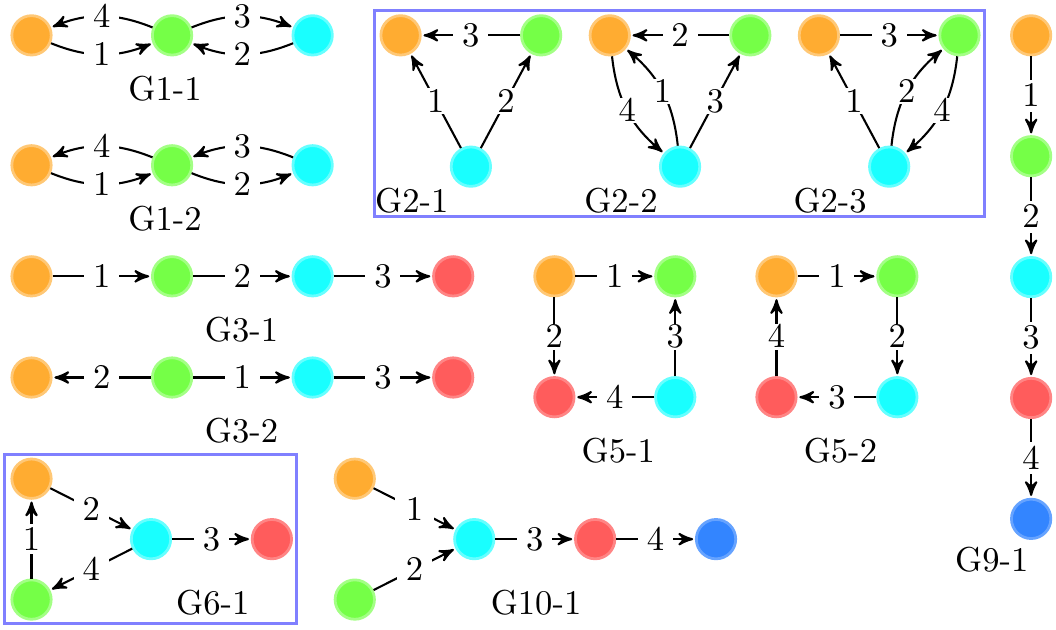}
    \caption{Motifs used in the experimental evaluation. The edge labels denote the edge order in $\sigma$ according to Definition \ref{defn:temporal-motif}.}
    \label{fig:motifs}
\end{figure}

\input{tables/tableapproxFull}
\input{tables/tableapproxFullReddit}

\subsection{Quality of Approximation}\label{subsec:approxfact}
We start by comparing the approximation provided by \algname\ with the current state-of-the-art sampling algorithms.
To allow for a fair comparison, since the algorithms we consider depend on parameters that are not directly comparable and that influence both the algorithm's running time and approximation, we run all the algorithms sequentially and measure the approximation they achieve by fixing their running time (i.e., we fix the parameters to obtain the same running time).

While \algname~can be trivially modified in order to work within a user-specified time limit, the same does not apply to \texttt{LS} or \texttt{ES}. Thus, to fix the same running time for all algorithms we devised the following procedure: i) for a fixed dataset and $\delta$, we found the parameters of \texttt{LS} such its runtime is at least one order of magnitude\footnote{We use such threshold since sampling algorithms are often required to run within a small fraction of time w.r.t.\ exact algorithms.} smaller than the exact algorithm by Mackey et al. \cite{mackey_chronological_2018} on motif G1-1 ii) we run \texttt{LS} on all the motifs on the fixed dataset with the parameters set as from i) and measure the time required for each motif; iii) we run \texttt{ES} such that it requires the same time\footnote{\texttt{ES} needs to preprocess each dataset before executing the sampling procedure. To account for this preprocessing step, we measured the time  for preprocessing as geometric average across ten runs, and added to \texttt{ES}'s running time such value divided by the number of motifs (i.e., 12).} as \texttt{LS}; iv) we ran \algnameone~and \algnametwo~with $c=1.25$ and a time limit that is the \emph{minimum} between the running time of \texttt{LS} and the one of \texttt{ES}. The parameters resulting from such procedure for all algorithms are in Supplement \ref{app:parameters}. 
Further discussion on how to set $c$ when running \algname\ is in Supplement \ref{app:cvalue}.

For each configuration, we ran each algorithm ten times, with a limit of 400 GB of (RAM) memory. We computed the so called Mean Average Percentage Error (MAPE) 
after removing the best and worst results, to remove potential outliers. 
MAPE is defined as follows: let $C_M'$  be the estimate of the count $C_M$ produced by one run of an algorithm, then  the relative error of such estimate is $|C_M' - C_M|/C_M$; the MAPE of the estimates is the average of the relative errors in percentage. In addition, we also computed the standard deviations of the estimates in output to the algorithms for each configuration, which will be reported to complement the mean values, small standard deviations correspond to more stable estimates of the algorithms around their MAPE values.

Table \ref{tab:approxSoBi} shows the MAPE values of the sampling algorithms on datasets SO, and BI.
On the SO dataset we observe that both versions of \algname\ provide more accurate estimates over current state-of-the-art sampling algorithms. \algnameone\ provides the best approximation on 6 out of 12 motifs, with \algnametwo\ providing the best results in all the other cases except for motif G2-1. Furthermore, both variants of \algname\ improve on all motifs w.r.t.\ \texttt{LS} and on 11 out of 12 motifs w.r.t.\ \texttt{ES}, with the error from \texttt{ES} being more than three times the error of \algname\ on such motifs. We note \algnameone\ and \algnametwo\ achieve similar results, supporting the idea that similar performances are obtained when the network's timestamps are distributed evenly, that is the case for such dataset (see Supplement \ref{app:edge_distribution}).

For the BI dataset the results are similar to the ones for the SO dataset. \algnameone\ achieves the lowest approximation error on 7 out of 12 motifs with \algnametwo\ achieving the lowest approximation error on all the other motifs but G2-2. \algnameone\ achieves better estimates than \texttt{LS} on all motifs and on 11 out of 12 motifs w.r.t.\ \texttt{ES}, \algnametwo\ behaves similarly, except that it improves over \texttt{LS} on 10 out of 12 motifs.

Table \ref{tab:approxRe} instead shows the results on the RE and EC datasets. On the former, \algnametwo\ achieves the lowest approximation error on 7 out of 12 motifs and it improves over \algnameone\ for most of the motifs. This is not surprising, since \cite{wang2020efficient} showed that RE has a more skewed edge distribution (see also the discussion in Supplement \ref{app:edge_distribution}), a scenario for which we expect \algnametwo\ to improve over \algnameone\ (see Section \ref{sec:RSRT}). Nonetheless, \algnameone\ achieves lower approximation error over 6 motifs when comparing to \texttt{LS} and on 10 motifs w.r.t.\ \texttt{ES}, while \algnametwo~achieves the best approximation error on most of the motifs, improving over \texttt{LS} on 8 motifs and over \texttt{ES} on 10 motifs.

Finally, we discuss the results on the EC dataset (again in Table \ref{tab:approxRe}), which is a 2.3 billion edges \emph{bipartite} temporal network and thus it does not contain the motifs marked with a rectangle in Figure \ref{fig:motifs}.
Note that the results for \texttt{ES} are missing, since \texttt{ES} did not complete any run with 400GB of memory on the motifs tested. (We discuss the high memory usage of \texttt{ES} in Section \ref{subsec:memory}).
On such dataset both variants of \algname\ perform better than \texttt{LS} on 7 out 8 motifs and achieve the lowest approximation error on 4 out 8 motifs each. Interestingly, the standard deviations of \algname\ are often similar or lower than the one of \texttt{LS}, which is another of the advantages of our algorithms. We observed similar trends, and even smaller variances for \algname\ across all the experiments on the other datasets (see Tables \ref{tab:approxSoBi}, \ref{tab:approxRe}).

These results, coupled with the theoretical guarantees of Section~\ref{sec:algs}, show that \algname\ outperforms the state-of-the-art sampling algorithms \texttt{LS} and \texttt{ES} for the estimation of most of the motif counts on temporal networks. 
Based on the results we discussed, \algnametwo\ seems also to usually provide similar estimates to \algnameone, 
while significantly improving over \algnameone\ when the network edges are not uniformly distributed such the \texttt{EC} dataset.
\begin{table}[t]
    \centering
    \caption{Minimum and peak memory usage in GB over all motifs in Figure \ref{fig:motifs}, \enquote{\ding{55}} denotes out of memory.}
    \label{tab:memorypeak}
        \begin{tabular}{*{5}{c}}
            \toprule
            \textbf{Dataset} & \algnameone & \algnametwo & \texttt{LS} & \texttt{ES}\\
            \midrule
            SO & 1.5 - 1.5 & 1.5 - 1.5 & 1.5 - 1.5 & 9.7 - 10.5\\
            BI & 3.5 - 3.5 & 3.5 - 3.7 & 3.5 - 4.2 & 28.8 - 29.5\\
            RE & 19.5 - 19.8 & 19.5 - 19.8 & 19.5 - 20.8 & 129.9 - 141.6\\
            EC & 71.0 - 71.0 & 71.0 - 71.0 & 71.0 - 71.0 & \ding{55}\\
            \bottomrule
        \end{tabular}
\end{table}

\subsection{Memory usage}\label{subsec:memory}
In this section we discuss the RAM usage of the various sampling algorithms. The results are in Table \ref{tab:memorypeak}, which shows for each algorithm the minimum and maximum amount of (RAM) memory used over all motifs in Figure \ref{fig:motifs} (collected on a single run). Both \algname's versions have lower memory requirements than \texttt{ES}, and equal or lower memory requirements than \texttt{LS}, on all datasets. \texttt{ES} ranges from requiring 6.5$\times$ up to 8.5$\times$ the memory used by \algname, making \texttt{ES} not practical for very large datasets such as EC (where \texttt{ES} did not completed any of the runs since it exceeded the 400GB of memory allowed for its execution). While \texttt{LS} and \algname\ have similar memory requirements, \texttt{LS} displays greater variance on BI and RC, which may be due to the fact that choosing windows with more edges, as done by \texttt{LS}, may require more memory to run the exact algorithm as subroutine.

\algname~is, thus, a memory efficient algorithm and hence a very practical tool to tackle the motif counting problem on temporal networks with billions edges. 

\section{Conclusions}
\label{sec:concl}
In this work we introduced \algname, a simple yet practical algorithm for the rigorous approximation of temporal motif counts. Our extensive experimental evaluation shows that \algname\ provides more accurate results and is much more scalable than the state-of-the-art sampling approaches.
There are several interesting directions for future research, including 
developing rigorous techniques for mining several motifs simultaneously and developing efficient algorithms to maintain estimates of motif counts in adversarial streaming scenarios.



\appendix
\newpage
\section*{Supplementary material.}


\section{Note on Bennett's Inequality} \label{app:bennett_note}
In this section we will discuss how to apply Bennett's inequality when the exact variance term, i.e.\ $v$, from Theorem \ref{theo:bennet} is unknown and we only have an upper bound $\hat{v}$ to such term (which is our case in practice).

Let us consider the function $f(v) = \alpha v \,h(\gamma/v)$ with $\alpha = s/B^2$, $\gamma = tB$ where $h(\cdot)$ is the same function $h(\cdot)$ appearing in Bennett's inequality (Theorem \ref{theo:bennet}). With a somehow tedious analysis one can prove that the function $f(v)$ is continuous and monotonic non increasing for every $v, \gamma > 0$ thus by definition, given $v_1 \ge v_2 > 0$ it holds $f(v_1)\le f(v_2)$. Observe that we can write the right-hand side of the Bennett's bound as $2\exp(-f(v))$, now is clear that given $v_1 \ge v_2 > 0$ it holds $2\exp(-f(v_2)) \le 2\exp(-f(v_1))$ by the property of negative exponential functions. This observation allows us to apply Bennett's inequality using an upper bound $\hat v$ to $v$. 

By the above observation we also get the intuition that a smaller gap between $\hat v$, the upper bound to the variance, and $v$ the actual variance will result in a tighter bound to the sample size $s$ (i.e., Bennett's inequality will provide more strict bounds). In our work we focused on a tradeoff between the sharpness of the bound $\hat v$ and its efficient computability in practice, since sharper bounds on $\hat v$ (than the one we provide) could be proved but they are much harder to compute and thus may not be of practical interest. We still leave for future works the space of improving the bounds in this work.

\section{\algnametwo\ - Theoretical Guarantees}\label{subsec:theoguaran_appendix}
In this section we will prove that, for an appropriate choice of the number $s$ of samples,
\algnametwo~is an $(\varepsilon, \eta)$-approximation algorithm. The analysis is similar to the one of \algnameone.

We define the indicator random variables, for each $u \in \mathcal{U}(M, \delta)$ and for $i=1, \dots,s$: $\tilde{X}_{u}^{i}=1 \text{ if $u \in \mathcal{S}_i$, 0 otherwise.}$
Each variable $\tilde{X}_u^i, i=1,\dots,s\, , u\in \mathcal{U}$ is a Bernoulli random variable for which it holds
\begin{equation}\label{eq:prob_x_u_i}
\mathbb{P}(\tilde{X}_{u}^{i} = 1) = \mathbb{P}(u \in \mathcal{S}_i) =
\frac{\tilde{r}_{u}}{\Delta_{T,2}} = \tilde{p}_u.
\end{equation}
With a similar analysis to the one for \algnameone, 
the following results hold.
\begin{Lemma}
    $ C_M' = \frac{1}{s} \sum_{i=1}^s X_i$ is an unbiased estimator for $C_M$, that is $\mathbb{E}[C_M']=C_M$. Where $X_i = \sum_{u \in \mathcal{U}}1/\tilde{p}_u \tilde{X}_{u}^{i}$.
\end{Lemma}
\begin{Lemma}
    \label{lem:alg2sample}
    \begin{displaymath}
    \mathrm{Var}\left(C_M'\right) = \mathrm{Var}\left(\frac{1}{s} \sum_{i=1}^{s} X_i\right) \leq \frac{C_M^2}{s}\left( \Delta_{T,2} - 1 \right)
    \end{displaymath}
\end{Lemma}

By applying Bennett's inequality (Theorem~\ref{theo:bennet}), we derive the following bound on the sample size $s$ which is also reported in the main text.
\begin{theorem} \label{theo:bound_s_app}
    Given $\varepsilon \in \mathbb{R}^{+}, \eta \in (0,1)$ let $X_1, \dots, X_s$ be the random variables associated with the counts at iteration $i=1,\dots,s$ of \algnametwo. If $s \geq  \frac{(\Delta_{T,2} -1)}{(1+\varepsilon) \ln(1+\varepsilon) - \varepsilon} \ln \left( \frac{2}{\eta} \right) $, then it holds
    \begin{displaymath}
    \mathbb{P}\left( \left|\frac{1}{s}\sum_{i=1}^s X_i - C_{M}\right| \geq \varepsilon C_{M} \right) \leq \eta
    \end{displaymath}
    that is, \algnametwo~is an $(\varepsilon, \eta)$-approximation algorithm.
\end{theorem}

\section{Description of the Equinix-Chicago Dataset}\label{app:EquinixDataset}
In this section, we briefly describe the Equinix-Chicago used in our experimental evaluation.

We built the Equinix-Chicago dataset\footnote{We used data available at \url{https://www.caida.org/data/passive/passive_2011_dataset.xml} publicly available under request.} starting from the internet packets collected by two monitors situated at Equinix data center, which captured the packets from Chicago to Seattle and vice versa on the 17 February 2011 for an observation period of 62 minutes.
Each edge $(x,y,t)$ represents a packet sent at time $t$, with microseconds precision, from the IPv4 address $x$ to the IPv4 destination $y$\footnote{We filtered packets which used different protocols, e.g., IPv6. Further, we did not assigned nodes at port level but at IP address level (e.g, let \texttt{addr} be an IPv4 address, then the edges corresponding to two packets sent at different ports: \texttt{addr}:80, \texttt{addr}:22 are mapped on the same endpoint node corresponding to \texttt{addr}).}. Such network is thus a \emph{bipartite} temporal network since edges, i.e., packet exchanges, occur only between nodes belonging to different cities, with no interactions captured among nodes of the same city. The final temporal network has more than 2.325 billion temporal edges and 11.157 million nodes. 

This dataset is publicly available at \url{https://github.com/VandinLab/PRESTO/}.

\section{Reproducibility - Links to External Resources}\label{app:links}
For reproducibility purposes, in this section we report the links to the datasets and implementations of \texttt{LS} and \texttt{ES} we used. 
\begin{itemize}
    \item Dataset SO is available at: \textcolor{black}{\url{http://snap.stanford.edu/temporal-motifs/data.html}};
    \item Datasets RE and BI are available at \textcolor{black}{\url{http://www.cs.cornell.edu/~arb/data/}};
    \item The current implementation of the \texttt{LS} algorithm is available at: \textcolor{black}{\url{https://gitlab.com/paul.liu.ubc/sampling-temporal-motifs}};
    \item The implementation of the \texttt{ES} algorithm is available at: \textcolor{black}{\url{https://github.com/jingjing-hnu/Temporal-Motif-Counting}}.
\end{itemize}

\section{Reproducibility - Parameters used in the experimental evaluation} \label{app:parameters}

We now discuss discuss the choice of the parameters resulting from our procedure for comparing the different sampling algorithms in Section \ref{subsec:approxfact}.

Note that our algorithms \algnameone\ and \algnametwo\ have only 1 tuning parameter in the experimental setting of Section \ref{subsec:approxfact}, since $\delta$ is provided in input (i.e., $c$, for which we also discuss how to set it to obtain the maximum efficiency in Supplement \ref{app:cvalue}). The same holds for \texttt{ES} which has the paramer $p$, that controls the fraction of sampled edges (see Section \ref{sec:relatedwork}). While \texttt{LS} has instead non trivial dependencies among the parameters $b$ and $r$ which are not entirely easy to set when it comes to different values of $c$ and $\delta$ also due to the lack of a qualitative theoretical analysis relating the parameters to the approximation quality of the results. 

Table \ref{tab:datasets_params} reports all the parameters used on the different datasets.
Given the procedure for running \algname's versions (described in Section \ref{subsec:approxfact}), their sample size $s$ is not deterministic, since only the execution time is fixed. In Table \ref{tab:datasets_params} we report the average sample sizes $s_{\text{\algnameone}}$ and $s_{\text{\algnametwo}}$, over all the runs on all the motifs on the given network, of \algnameone\ and \algnametwo\ respectively. Recall that the values of $p$ for \texttt{ES} where adapted to each motif on each configuration following the approach described in Section \ref{subsec:approxfact} 
selecting such values from a grid of values that is, we selected the best $p$ in order to have a final running time of \texttt{ES} to be close to the one of \texttt{LS} choosing $p$ from the following grid $[10^{-6}, 5\cdot 10^{-6}, 10^{-5}, 5\cdot10^{-5}, 10^{-4}, 5\cdot 10^{-4}, 10^{-3}, 2 \cdot 10^{-3}, 5 \cdot 10^{-3}, 10^{-2} ]$. Thus, we report in Tab.\ \ref{tab:datasets_params} the range from which $p$ was chosen in practice as result of our procedure for fixing the parameters.

\begin{table}[t]
    \centering
    \caption{Parameters used in the comparison of \texttt{LS}, \texttt{ES}, \algnameone\ and \algnametwo. For the parameter $p$ we report the range since it's value depends on the specific motif (see Section \ref{subsec:approxfact}). We further discuss in Supplement \ref{app:cvalue} the motivation for our choice of $c$.}
    \label{tab:datasets_params}
    \scalebox{0.9}{
        \begin{tabular}{cccccccc}
            \toprule
            Dataset&$c$&$\delta$& $r$ & $b$ &$p$ &$s_{\text{\algnameone}}$ & $s_{\text{\algnametwo}}$ \\
            \midrule
            SO & 1.25 & 86400 & 30 & 5 & $10^{-[5,3]}$ &230 & 158\\
            BI & 1.25 & 3200 & 175 & 5 & $10^{-[5,3]}$ &3311 & 905\\
            RE & 1.25 & 3200 & 400 &  5& $10^{-[5,3]}$ &9378 & 2061\\
            EC & 1.25 & 100000 & 250 & 5 & $10^{-[5,3]}$ &1530 & 1168\\
            \bottomrule
        \end{tabular}
    }
\end{table}
  
\section{On the choice of $c$} \label{app:cvalue}
In this section we briefly discuss how to choose the value of $c$ when running \algname.
 
The major advantages of \algname, come from the efficiency in time, small memory usage and from the scalability it achieves, which coupled with the theoretical guarantees provided and the precise estimates achieved in practice, make our algorithm a fundamental tool when analysing large temporal networks, it is thus important to set \algname's parameter (i.e., $c$) correctly to exploit such features at their best. Our versions of \algname\ have only the parameter $c$ to be setted, since $\delta$ is part of the problem's input. We now discuss the importance of keeping $c$ small (e.g., $c <2$). To do this, we consider a concrete example to show that a small $c$ leads to more efficiency and scalability. Suppose we want to approximate the count of motif G5-1 from Figure \ref{fig:motifs} 
on EC dataset\footnote{This is an arbitrary example, similar results can be observed on almost all the motifs and datasets.} and we want to cover with samples a fixed length of the interval $[t_1,t_m]$ over the $s$ iterations, i.e., $\sum_{i=1}^s |t_{r}^i + c \delta - t_r^i| = sc \delta = V$ where $[t_r^i, t_r^i + c \delta]$ is the temporal window sampled by \algnameone\footnote{Similar observations hold for \algnametwo.} at iteration $i=1,\dots,s$ and $V$ is the fixed value, $V$ can be interpreted as the summation of the durations of each sample ($c\delta$) over all the samples (i.e., how much time of $[t_1,t_m]$ over the $s$ iterations \algname\ examines). Fixed $V$ and $\delta$, one can explore different values of $c$ which thus let only one possibility to set $s$, by solving the equation $s = V/c\delta$. Table \ref{tab:exampleV} presents the running times for various values of $c$ with fixed $\delta = 10^5$, $V = 1.25 \cdot 10 ^{8}$. When $c$ decreases there is a substantial decrease in running time, in fact with $c=10$ the running time is more than $3.9\times$ the running time with $c=1.25$. This should not surprise the reader, since such results is in accordance with our complexity analysis (see Section \ref{subsec:complexity_analysis}) which shows that larger samples in terms of edges (as obtained by sampling windows with a higher $c$) result in a higher running time. 

\begin{table}[t]
    \centering
    \caption{Running times (geometric mean of five different runs) of \algnameone~using $8$ threads on motif G5-1 on EC dataset.}
    \label{tab:exampleV}
    \scalebox{1}{
        \begin{tabular}{ccc}
            \toprule
            $c$& $s$ & running time (sec)\\
            \midrule
            10 & 125 & 242.09 \\
            5 & 250 & 128.94 \\
            2.5 & 500 & 76.43\\      
            1.25 & 1000 & 61.41\\
            \bottomrule
        \end{tabular}
    }
\end{table}

Thus the best choice in order to set $c$ is to keep $c<2$ based on all the experiments we performed, to exploit efficiency and scalability at the maximum of their level.

\section{Running time comparison} \label{app:runningtime}
\begin{figure}[t]
    \centering
    \subfloat[SO]{
        \includegraphics[width=.35\linewidth]{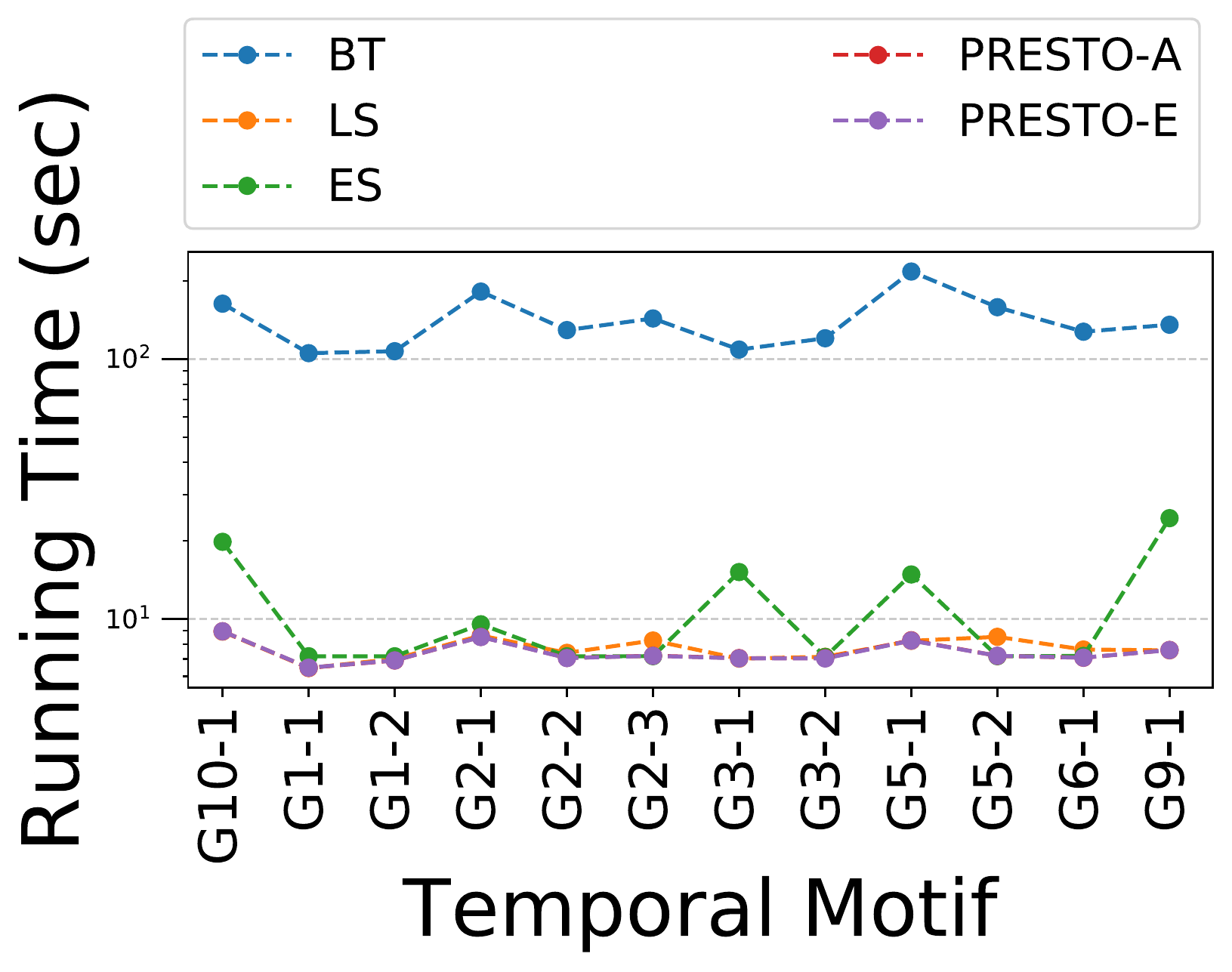}
        \label{subfig:timeSO}
    }
    \subfloat[BI]{
        \includegraphics[width =.35\linewidth]{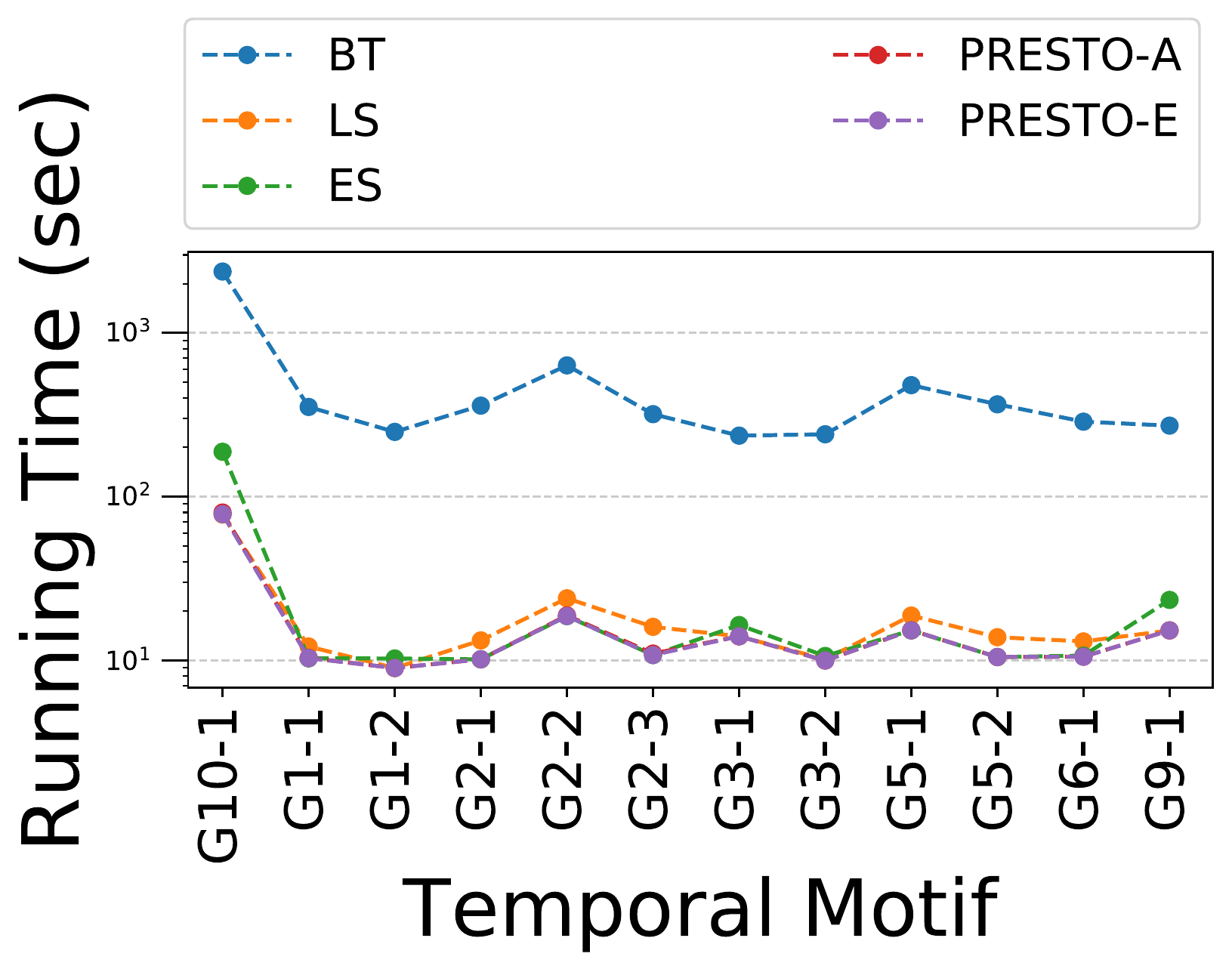}
        \label{subfig:timeBI}
    }\\
    \subfloat[RE]{
        \includegraphics[width=.35\linewidth]{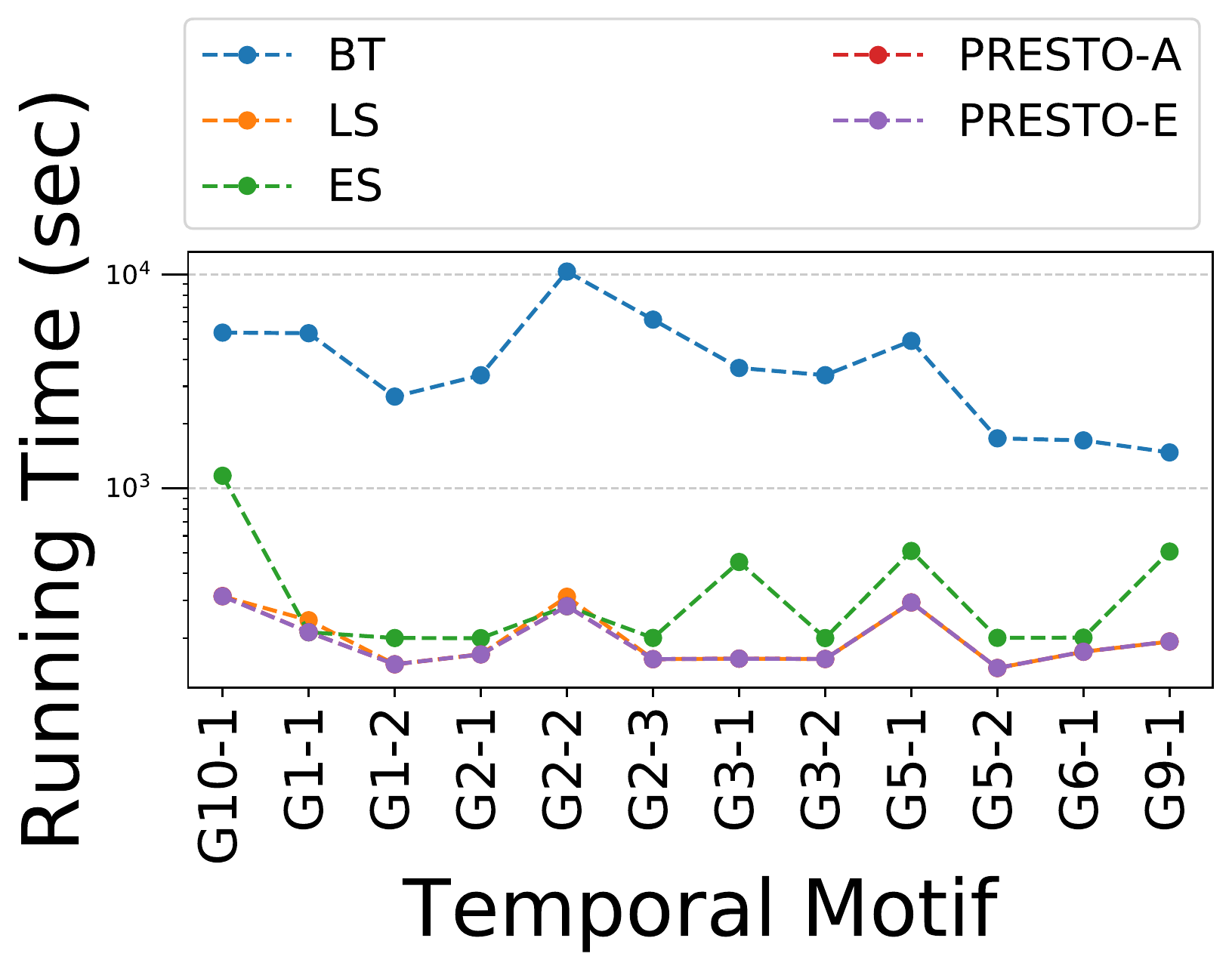}
        \label{subfig:timeRE}
    }
    \subfloat[EC]{
        \includegraphics[width =.35\linewidth]{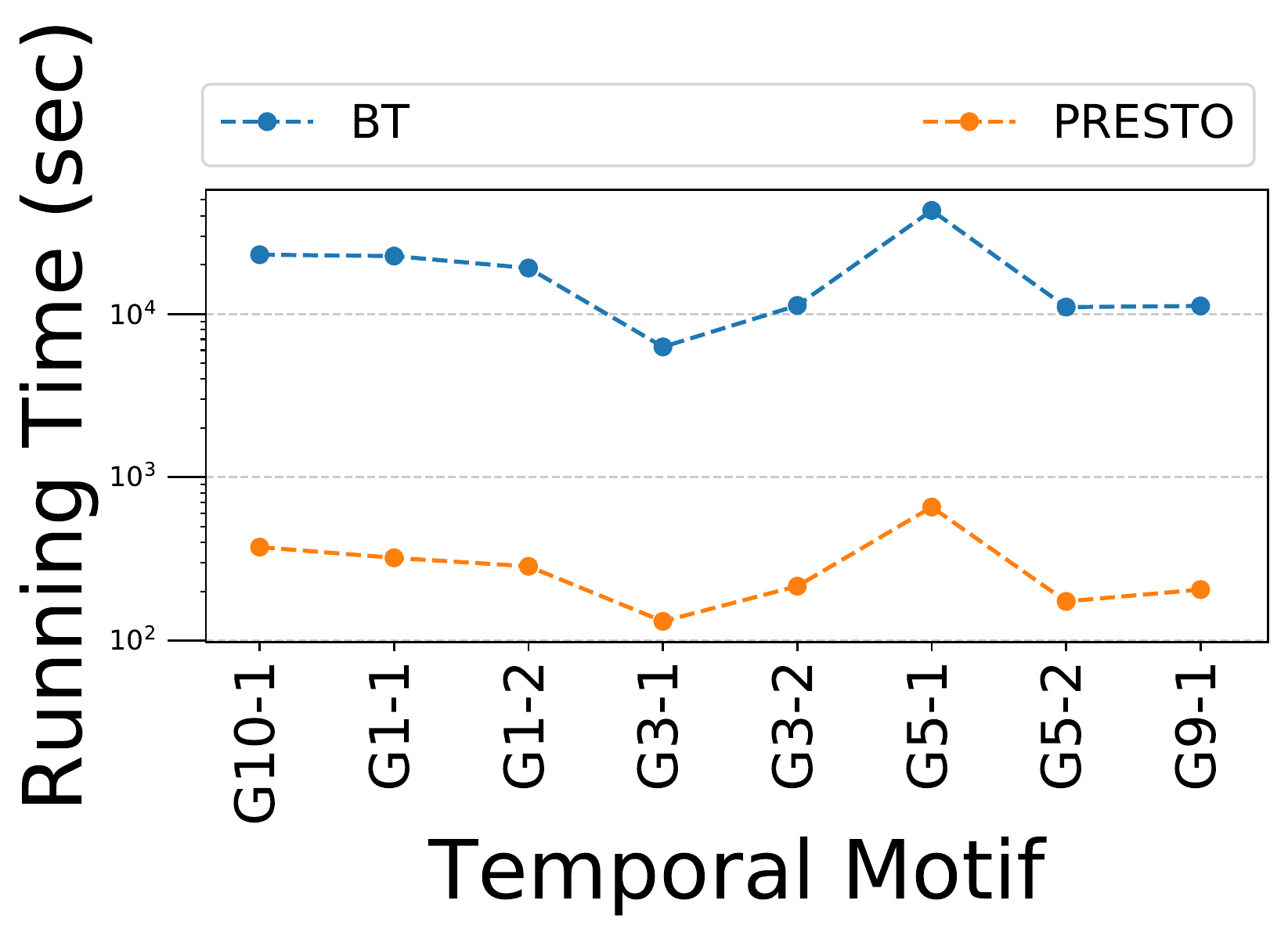}
        \label{subfig:timeEC}
    }
    \caption{Running times on the different datasets, according to the procedure in Section \ref{subsec:expsetup}.}
    \label{fig:runningtimes_app}
\end{figure}
In this section we discuss the running times under which we obtained the results in Section \ref{subsec:approxfact}.

We measured the running time of \algnameone, \algnametwo, \texttt{LS}, and \texttt{ES}  over all experiments (i.e., datasets and motifs) described in Section \ref{subsec:approxfact}, 
we also measured the running time of \texttt{BT} over three runs on the same configurations (i.e., dataset ad $\delta$). For all algorithms we then computed the geometric average of their running times for each pair (dataset, motif). The results are shown in 
Figure \ref{fig:runningtimes_app}. As expected, \algname\ runs with at least an order of magnitude less than \texttt{BT}. The reader may appreciate how the results in Section \ref{subsec:approxfact} 
were obtained with \algname's running time bounded by the minimum running time among \texttt{ES} and \texttt{LS} (fixed the motif and the dataset). Furthermore we note that on several motifs (G10-1, G9-1) on almost all datasets, even if we set $p$ small, i.e., $10^{-5}$, \texttt{ES} runs within a higher time than \texttt{LS} and thus than \algname, we believe that this may be due to \texttt{ES}'s procedure to explore the search space which may not suite such motifs. To this end we find that \algname\ is not affected by the motif topology.
We present the results on \algname's running time against \texttt{BT} on the EC dataset in  
Figure (\ref{fig:runningtimes_app}d) which corresponds to the times under which we obtained the results in Table \ref{tab:approxcaida}. 
Note that \algname\ runs orders of magnitude faster than \texttt{BT} while still providing accurate estimates. For example on G3-2 \algname\ achieves a 3\% relative error (see Table \ref{tab:approxRe}) while being 53$\times$ faster than \texttt{BT}.

Coupled with the results of Sections \ref{subsec:approxfact} 
and \ref{subsec:memory}, these aspects show that \algname\ is an efficient algorithm to obtain accurate estimates of motif counts from large temporal networks, especially when the running time to obtain an estimate is very limited, this may be the case on very large datasets.

\section{Scalability of \algname's parallel implementation}\label{app:scalability_suppl}
\begin{figure}[h]
    \centering
    \begin{tabular}{cc}
        \includegraphics[width=.35\linewidth]{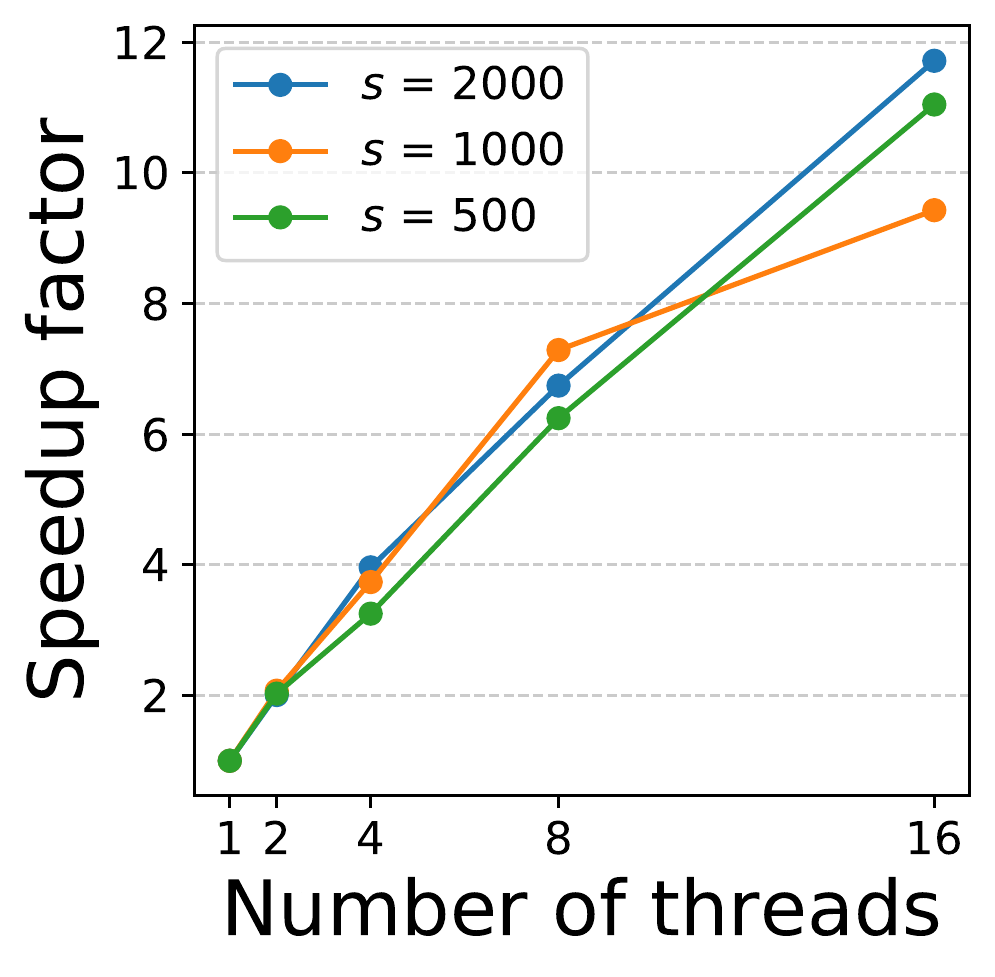} & 
        \includegraphics[width=.35\linewidth]{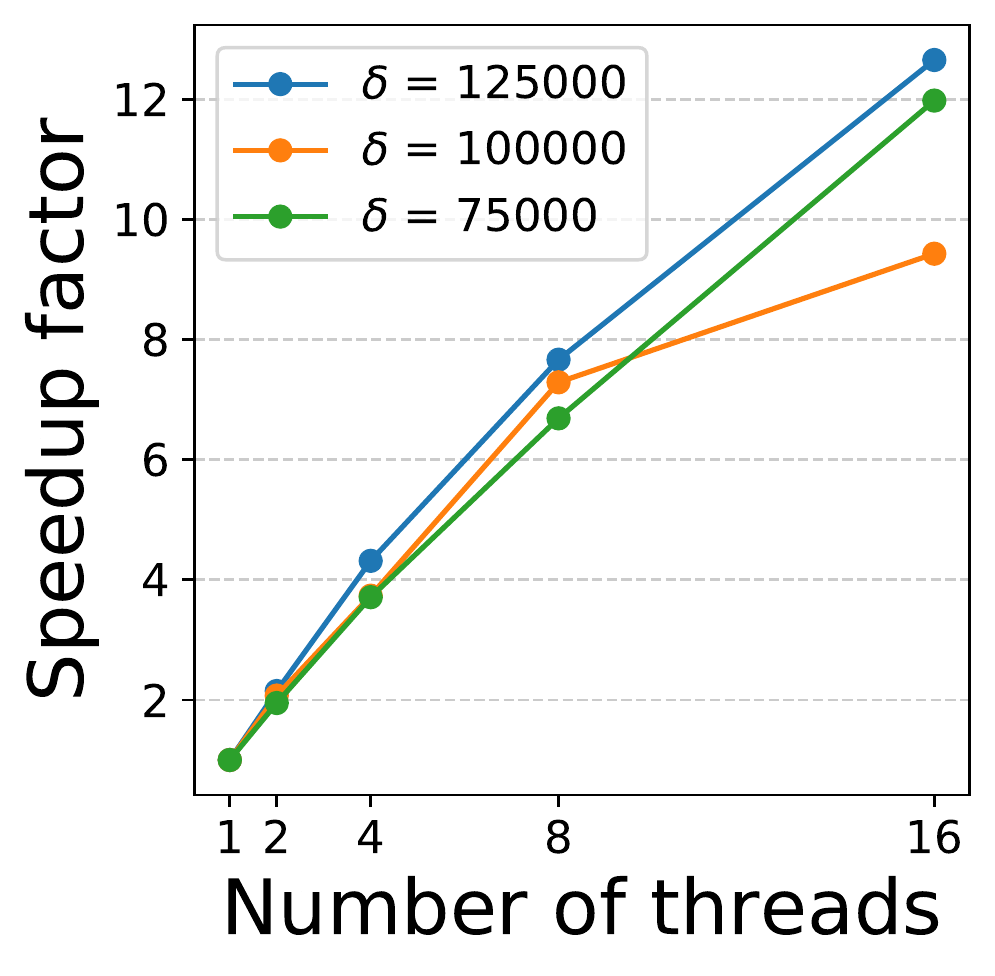} 
    \end{tabular}
    \caption{Speed-up factor achieved through the parallel implementation of \algname\ by varying the number of threads, when using different values for the sample size $s$ (left) and of the temporal duration $\delta$ (right).}
    \label{fig:paralleldelta}
\end{figure}
As discussed in Section \ref{subsec:generalSchema}, \algname~can be easily parallelized by executing the $s$ iterations of Algorithm \ref{alg:generalschema} in parallel, in this section we discuss such implementation. We implemented such strategy through a \emph{thread pooling} design pattern. We present the results of the multithread version of \algnameone. Results for \algnametwo\ are similar. 

We present the results on EC for motif G5-1 (see Figure \ref{fig:motifs}), which has the highest running time over all the experiments in Section \ref{subsec:approxfact} (see Figure (\ref{fig:runningtimes_app}d)), similar results are observed on other motifs and datasets. Figure \ref{fig:paralleldelta} shows the speed-up $T_1/T_j$ where $T_j, j\in\{1,2,4,8,16\}$ is the geometric average of the running time over five runs (with fixed parameters) executing \algnameone\ with $j$ threads. In Figure \ref{fig:paralleldelta} (left) we fixed $c$ and $\delta$ as in Supplement \ref{app:parameters} (Table \ref{tab:datasets_params}), 
and we show how the speed-up changes by varying the sample size $s$. 
We observe a nearly linear speed-up for up to 8 threads, and a speed-up for 16 threads that is slightly less than linear. 
In Figure \ref{fig:paralleldelta} (right) we fix $c=1.25$ and $s=1000$ and we show how the speed-up changes across different values of $\delta \in \{75,100,125\} \cdot 10^{3}$. As we can see there is a linear and slightly superlinear behaviour up to 8 threads, with the performances becoming slightly less than linear when using 16 threads, similarly to the results obtained for the tests with fixed $\delta$. This behaviour may be due to design pattern or to the time needed to process each sample, intuitively if the time to process each sample is not large, then there may be much time spent by the algorithm in synchronization when using a large number ($\ge$ 16) of threads and this may be the case on EC dataset.

All together, the results of our experimental evaluation show that \algname\ obtains up to $500\times$ of speed-up over exact procedure on the tests we performed while providing accurate estimates through its parallel implementation (53$\times$ faster through its sequential implementation, while its parallel implementation with 16 threads executes with a speedup of more than 10$\times$ over the sequential one), enabling the efficient analysis of billion edges networks.

\section{Edge Timestamps distribution - Skewed (Dataset RE) vs Uniform (Dataset SO)}\label{app:edge_distribution}
In this section we compare the edge distributions, with respect to the timestamps, of the datasets SO and RE used in our experimental evaluation in Section \ref{sec:experimentaleval}. To obtain such distribution we computed for each timestamp $t_i, i=1,\dots, m$ s.t.\ $\exists (u,v,t_i)\in E$, the number of edges in the forward window of length $c\delta$ starting from $t_i$, i.e., we computed $|\{(x,v,t)\in E | t \in[t_i, t_i+c \delta]  \}|$ for each $t_i, i=1,\dots,m$, where we kept the parameters $c$ and $\delta$ as in Table \ref{tab:datasets_params}. This corresponds, for example, to the number of edges in each sample candidate of \algnametwo. The results are shown in Figure \ref{fig:tdim}. 

\begin{figure}[t]
    \centering
    \begin{tabular}{cc}
        \includegraphics[width=.47\linewidth]{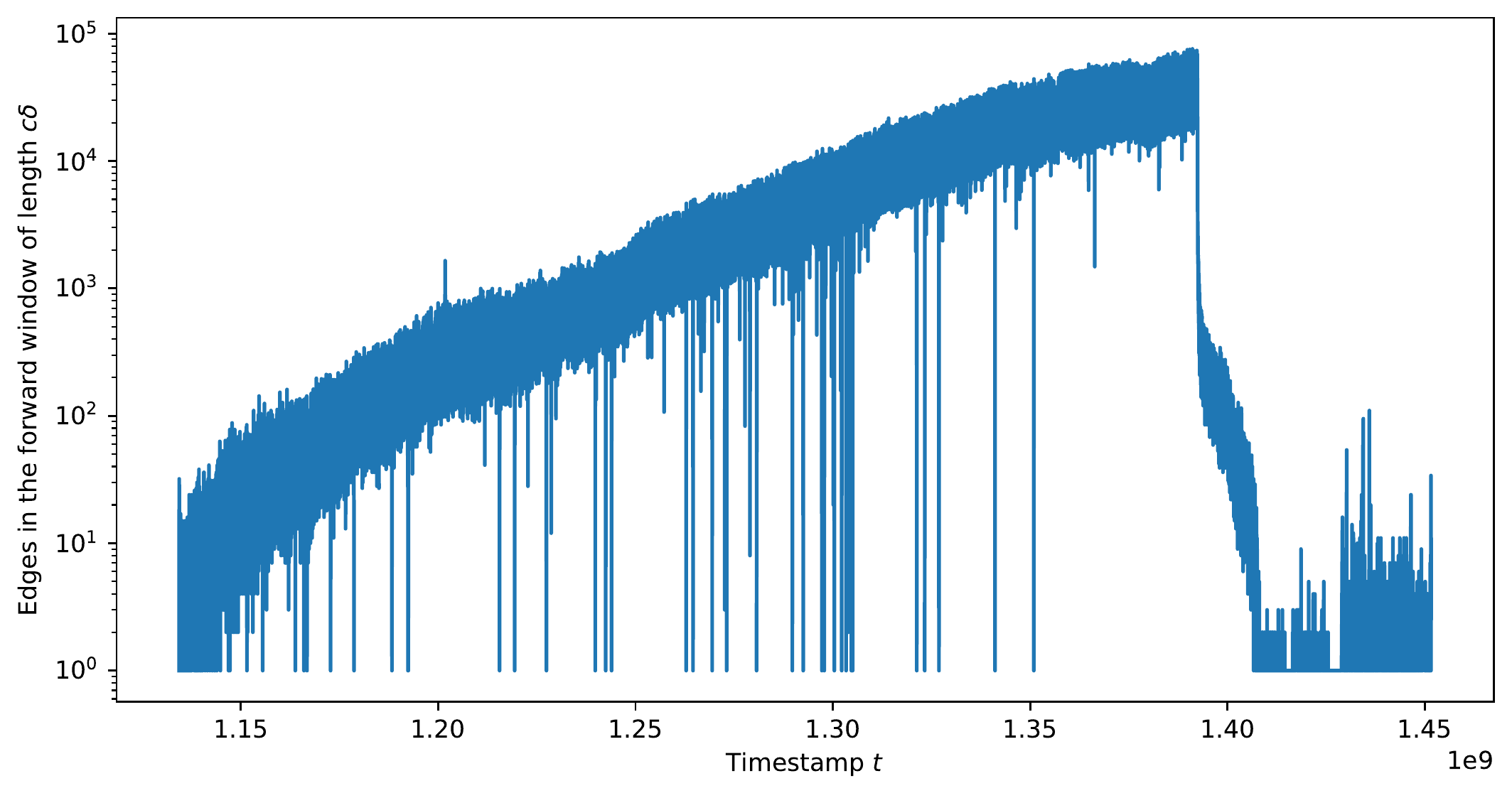} &
        \includegraphics[width=.47\linewidth]{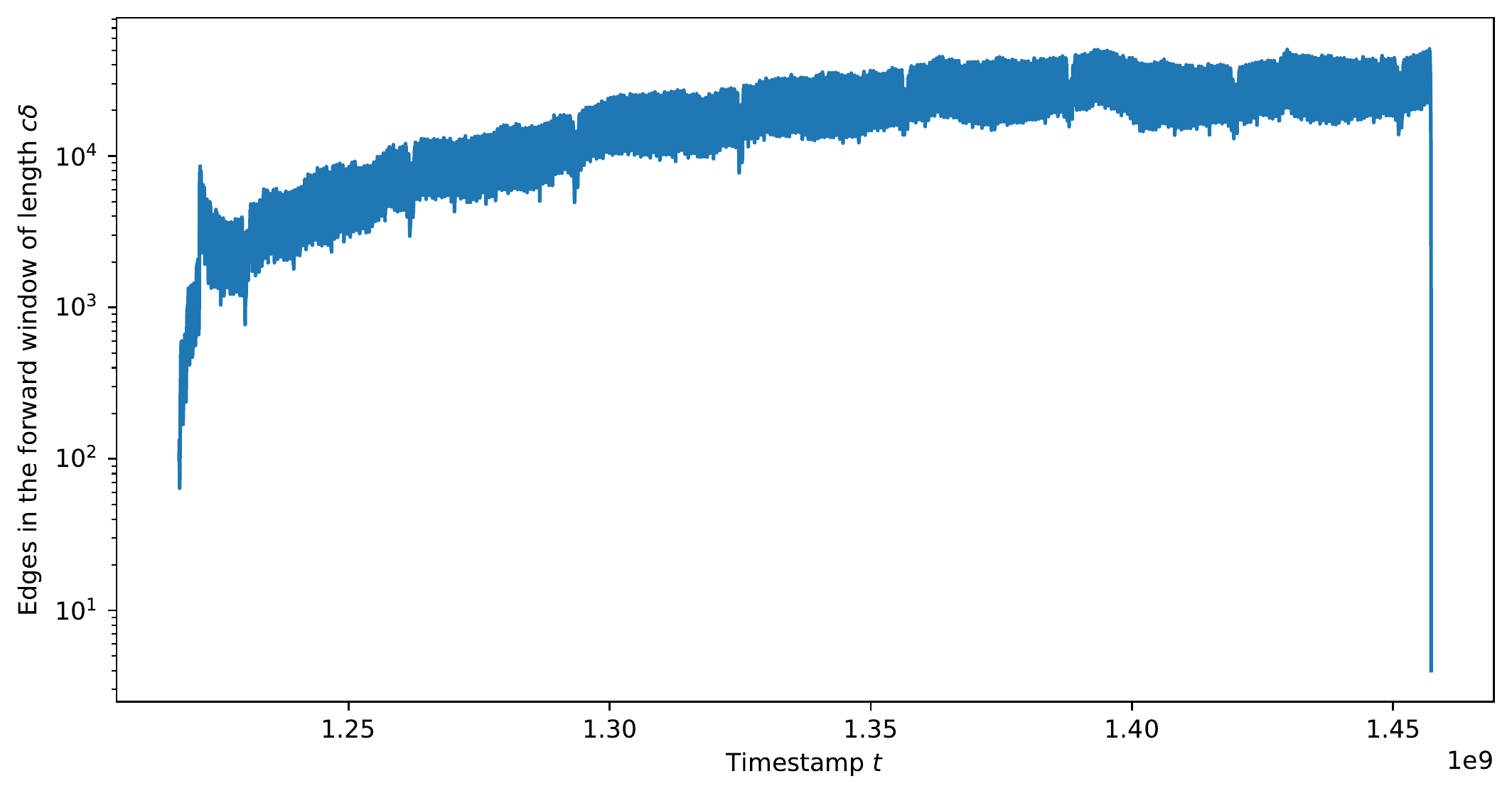} \\
    \end{tabular}
    \caption{(Left): distribution of the edges on dataset RE. (Right): distribution of the edges on dataset SO. Both plot's y-axis are in \emph{log} scale.}
    \label{fig:tdim}
\end{figure}

We first discuss the distribution of the timestamps of the edges in the RE dataset (Figure \ref{fig:tdim} (left)), we observe immediately that such distribution is very skewed. In particular, the number of edges in the windows of length $c\delta$ starting from edges with timestamps in the first quarter of $[t_1,t_m]$, i.e., $t < 1.25\cdot 10^{9}$, do not exceed $10^3$ edges. While most of the temporal windows starting at timestamps in the middle interval of $[t_1,t_m]$, i.e., $1.25\cdot 10^{9} <t< 1.4\cdot 10^{9}$, have much more edges ranging from $10^3$ to $10^5$ (several orders of magnitude edges more than windows from the first quarter of the network's timespan). In the last quarter of $[t_1,t_m]$ instead ($t> 1.4\cdot 10^9$) the dataset presents very sparse edges with most of the windows (which are sample candidates in \algname) not exceeding $10^2$ edges, note that the windows starting in the middle section of $[t_1,t_m]$ have at least one order of magnitude additional edges (i.e., $[10^3,10^5]$ edges). This shows how the timestamps of the edges in RE have a very skewed distribution.

Figure \ref{fig:tdim} (right) instead shows the distribution of the timestamps of the edges on the dataset SO. As we can see, except for a very brief transient state, the timestamps are almost uniformly distributed on $[t_1,t_m]$, with the windows of length $c\delta$ centered at each timestamp $t_i, i=1,\dots,m$ having almost the same number of edges, or ranging in at most one order of magnitude of difference. Therefore the edges of $E$ present a very uniform distribution of the timestamps in the dataset SO over the interval $[t_1,t_m]$. 

We  believe that such distributions may affect the results of the approximation of the different version of \algname\ as discussed in the main text, with \algnametwo\ being less sensible to skewed behaviours than \algnameone.

\end{document}

%% file: tables/tableapproxFull.tex
\begin{table*}[t]
        \centering
        \caption{Comparison of the relative approximation error and standard deviation on datasets SO and BI (see Section~\ref{subsec:expsetup}). For each combination of dataset and motif we highlight the algorithm with minimum MAPE value.}
        \label{tab:approxSoBi}
        \scalebox{0.75}{
        \begin{tabular}{*{11}{c}}
            \toprule
            & \multicolumn{5}{c}{SO} & \multicolumn{5}{c}{BI}\\ 
            \cmidrule(lr){2-6}
            \cmidrule(lr){7-11}
            &  &\multicolumn{4}{c}{Approximation Error}& &\multicolumn{4}{c}{Approximation Error}\\ 
            \cmidrule(lr){3-6}
            \cmidrule(lr){8-11}
            &  &\multicolumn{2}{c}{\algname} & & & &\multicolumn{2}{c}{\algname} & &\\ 
            \cmidrule(lr){3-4}
            \cmidrule(lr){8-9}
            \textbf{Motif} & $C_M$ & \texttt{A} & \texttt{E} & \texttt{LS} &\texttt{ES} & $C_M$ & \texttt{A} & \texttt{E} & \texttt{LS} &\texttt{ES}\\ 
            \midrule
            G10-1 & 7.8$\cdot 10^{8}$ & \textbf{4.1$\pm$2.4\%} & 4.2$\pm$2.0\% & 4.2$\pm$2.9\%  & 12.9$\pm$9.3\%&   1.2$\cdot 10^{10}$ & 14.2$\pm$8.7\% & \textbf{13.1$\pm$6.1\%} & 25.8$\pm$10.2\%  & 16.9$\pm$16.9\%\\
            G1-1 & 2.1$\cdot 10^{6}$ & 3.3$\pm$0.8\% & \textbf{1.9$\pm$0.8\%} & 7.8$\pm$3.5\%  & 23.3$\pm$11.3\%&  6.8$\cdot 10^{7}$ & \textbf{7.1$\pm$4.2\%} & 15.8$\pm$9.2\% & 21.2$\pm$12.8\%  & 44.9$\pm$30.6\%\\
            G1-2 & 8.3$\cdot 10^{5}$ & 6.6$\pm$2.5\% & \textbf{2.9$\pm$2.7\%} & 8.7$\pm$4.4\%  & 62.0$\pm$19.0\%&  5.7$\cdot 10^{7}$ & \textbf{9.3$\pm$4.1\%} & 15.5$\pm$10.0\% & 25.1$\pm$15.6\%  & 59.7$\pm$29.4\%\\
            G2-1 & 7.7$\cdot 10^{5}$ & 4.4$\pm$1.6\% & 3.9$\pm$2.4\% & 8.6$\pm$4.3\%  & \textbf{1.6$\pm$0.9\%}&    2.3$\cdot 10^{6}$ & \textbf{5.0$\pm$3.0\%} & 7.1$\pm$3.9\% & 21.4$\pm$7.3\%  & 30.5$\pm$15.3\%\\
            G2-2 & 3.5$\cdot 10^{5}$ & \textbf{3.2$\pm$3.0\%} & 4.6$\pm$3.2\% & 9.7$\pm$2.8\%  & 53.4$\pm$16.4\%&  4.9$\cdot 10^{5}$ & 11.0$\pm$3.7\% & 26.8$\pm$10.9\% & 34.6$\pm$26.1\%  & \textbf{3.4$\pm$2.0\%}\\
            G2-3 & 7.0$\cdot 10^{5}$ & \textbf{9.3$\pm$6.8\%} & 9.4$\pm$5.1\% & 10.1$\pm$5.5\%  & 60.7$\pm$28.4\%& 9.9$\cdot 10^{5}$ & \textbf{7.9$\pm$4.9\%} & 22.7$\pm$6.6\% & 19.6$\pm$10.2\%  & 46.9$\pm$12.4\%\\
            G3-1 & 2.3$\cdot 10^{8}$ & \textbf{1.6$\pm$0.9\%} & 2.7$\pm$1.4\% & 4.6$\pm$2.8\%  & 9.8$\pm$4.8\%&    7.3$\cdot 10^{8}$ & 2.5$\pm$1.2\% & \textbf{2.2$\pm$1.9\%} & 21.9$\pm$2.5\%  & 13.7$\pm$7.0\%\\
            G3-2 & 2.3$\cdot 10^{8}$ & 2.9$\pm$1.8\% & \textbf{1.8$\pm$1.3\%} & 5.3$\pm$2.8\%  & 11.3$\pm$7.5\%&   5.7$\cdot 10^{8}$ & 3.6$\pm$1.4\% & \textbf{1.7$\pm$1.2\%} & 21.8$\pm$3.1\%  & 17.6$\pm$8.5\%\\
            G5-1 & 8.0$\cdot 10^{5}$ & \textbf{4.0$\pm$2.1\%} & 4.2$\pm$1.6\% & 5.0$\pm$2.9\%  & 22.3$\pm$7.0\%&   5.1$\cdot 10^{6}$ & \textbf{14.5$\pm$8.9\%} & 17.6$\pm$4.4\% & 16.9$\pm$9.3\%  & 32.4$\pm$18.8\%\\
            G5-2 & 5.0$\cdot 10^{5}$ & \textbf{4.2$\pm$2.4\%} & 5.0$\pm$2.4\% & 5.4$\pm$3.0\%  & 34.4$\pm$15.9\%&  1.3$\cdot 10^{6}$ & 23.2$\pm$9.0\% & \textbf{17.6$\pm$4.8\%} & 23.3$\pm$21.9\%  & 33.8$\pm$19.2\%\\
            G6-1 & 2.0$\cdot 10^{6}$ & 6.4$\pm$4.0\% & \textbf{6.1$\pm$4.1\%} & 12.7$\pm$3.6\%  & 47.6$\pm$11.5\%& 1.0$\cdot 10^{7}$ & \textbf{7.0$\pm$4.6\%} & 16.7$\pm$9.4\% & 19.4$\pm$14.3\%  & 37.3$\pm$13.4\%\\
            G9-1 & 4.1$\cdot 10^{8}$ & 4.4$\pm$1.7\% & \textbf{2.5$\pm$0.7\%} & 6.0$\pm$2.9\%  & 10.5$\pm$8.3\%&   5.8$\cdot 10^{8}$ & \textbf{2.3$\pm$2.3\%} & 3.3$\pm$1.6\% & 26.7$\pm$8.1\%  & 11.8$\pm$7.1\%\\
            \bottomrule
        \end{tabular}
}
\end{table*}

%% file: tables/tableapproxFullReddit.tex
\begin{table}[t]
        \centering
        \caption{Approximation results on datasets RE and EC.}
        \label{tab:approxRe}
        \scalebox{0.75}{
        \begin{tabular}{*{10}{c}}
            \toprule
            & \multicolumn{5}{c}{RE} & \multicolumn{4}{c}{EC}\\ 
            \cmidrule(lr){2-6}
            \cmidrule(lr){7-10}
            &  &\multicolumn{4}{c}{Approximation Error}& &\multicolumn{3}{c}{Approximation Error}\\ 
            \cmidrule(lr){3-6}
            \cmidrule(lr){8-10}
            &  &\multicolumn{2}{c}{\algname} & & & &\multicolumn{2}{c}{\algname} &\\ 
            \cmidrule(lr){3-4}
            \cmidrule(lr){8-9}
            \textbf{Motif} & $C_M$ & \texttt{A} & \texttt{E} & \texttt{LS} &\texttt{ES} & $C_M$ & \texttt{A} & \texttt{E} & \texttt{LS}\\ 
            \midrule
            G10-1  & 5.8$\cdot 10^{9}$ & 14.5$\pm$6.7\% & \textbf{9.6$\pm$6.9\%} & 14.1$\pm$4.4\%  & 32.8$\pm$13.8\%&  6.3$\cdot 10^{10}$ & \textbf{2.8$\pm$2.8\%} & 4.2$\pm$3.2\% & 4.7$\pm$1.9\%\\
            G1-1  & 1.5$\cdot 10^{8}$ & 25.3$\pm$9.6\% & 25.2$\pm$15.2\% & 30.7$\pm$13.0\%  & \textbf{17.1$\pm$6.5\%}&  1.5$\cdot 10^{11}$ & 7.6$\pm$4.1\% & \textbf{5.1$\pm$3.4\%} & 8.0$\pm$3.7\%\\
            G1-2  & 1.2$\cdot 10^{8}$ & 43.8$\pm$15.5\% & \textbf{24.6$\pm$8.2\%} & 31.3$\pm$20.9\%  & 65.3$\pm$10.2\%& 1.5$\cdot 10^{11}$ & 6.5$\pm$4.1\% & \textbf{4.6$\pm$2.4\%} & 7.8$\pm$3.0\%\\
            G2-1  & 3.3$\cdot 10^{7}$ & 7.4$\pm$2.7\% & \textbf{4.3$\pm$2.7\%} & 8.4$\pm$5.3\%  & 20.1$\pm$9.5\%& -&- &- &-\\
            G2-2  & 2.9$\cdot 10^{7}$ & 39.4$\pm$23.6\% & 46.1$\pm$11.9\% & 33.4$\pm$13.4\%  & \textbf{9.8$\pm$4.7\%}& -&- &- &-\\
            G2-3  & 9.4$\cdot 10^{7}$ & 18.4$\pm$7.3\% & 16.9$\pm$7.6\% & \textbf{16.7$\pm$5.3\%}  & 32.6$\pm$14.4\%& -&- &- &-\\
            G3-1  & 4.2$\cdot 10^{8}$ & 3.1$\pm$1.7\% & \textbf{1.8$\pm$1.1\%} & 3.3$\pm$1.2\%  & 10.2$\pm$4.7\%& 1.6$\cdot 10^{10}$ & \textbf{1.9$\pm$1.5\%} & 3.5$\pm$1.8\% & 3.0$\pm$1.8\%\\
            G3-2  & 5.3$\cdot 10^{8}$ & 2.7$\pm$1.3\% & \textbf{1.2$\pm$1.1\%} & 3.1$\pm$1.3\%  & 7.5$\pm$3.3\%& 1.6$\cdot 10^{10}$ & \textbf{2.6$\pm$1.0\%} & 2.8$\pm$2.0\% & 3.7$\pm$2.2\%\\
            G5-1  & 8.2$\cdot 10^{7}$ & 8.1$\pm$2.6\% & \textbf{7.0$\pm$3.8\%} & 8.2$\pm$6.0\%  & 31.9$\pm$9.8\%& 3.4$\cdot 10^{11}$ & 10.4$\pm$7.0\% & \textbf{8.2$\pm$4.6\%} & 14.9$\pm$5.4\%\\
            G5-2  & 1.7$\cdot 10^{7}$ & 28.1$\pm$6.3\% & \textbf{13.8$\pm$6.6\%} & 25.9$\pm$10.9\%  & 52.0$\pm$15.7\%& 3.8$\cdot 10^{10}$ & 14.2$\pm$5.4\% & \textbf{4.5$\pm$4.0\%} & 9.0$\pm$5.0\%\\
            G6-1  & 5.5$\cdot 10^{7}$ & 33.2$\pm$17.6\% & 30.7$\pm$7.3\% & \textbf{20.0$\pm$9.7\%}  & 50.0$\pm$16.8\%& -&- &- & -\\
            G9-1  & 6.0$\cdot 10^{8}$ & \textbf{3.8$\pm$2.6\%} & 5.4$\pm$2.0\% & 4.5$\pm$2.8\%  & 19.9$\pm$18.8\%& 8.7$\cdot 10^{10}$ & \textbf{6.3$\pm$3.8\%} & 6.8$\pm$2.2\% & 9.9$\pm$4.6\%\\
            \bottomrule
        \end{tabular}
}
\end{table}